\documentclass[lettersize,journal]{IEEEtran}
\usepackage{amsmath,amsfonts}
\usepackage{amssymb}
\newtheorem{theorem}{Theorem}
\usepackage{algorithm}
\usepackage{algpseudocode}
\usepackage{array}
\usepackage{graphicx}
\usepackage{booktabs}
\usepackage[caption=false,font=footnotesize]{subfig}
\usepackage{textcomp}
\usepackage{stfloats}
\usepackage{url}
\usepackage{verbatim}
\usepackage{graphicx}
\usepackage{cite}
\begin{document}

\title{Stability-Performance Trade-off in Maritime Integrated Sensing, Communication, and Computation: A Lyapunov-Guided Dual-Timescale Hierarchical RL Approach}

\author{
	Ruoguang Li, \IEEEmembership{Senior Member, IEEE}, Yu Wang, Bingxin Tian, 
Benedetta Picano, \IEEEmembership{Member, IEEE},\\
Yingyang Chen, \IEEEmembership{Senior Member, IEEE}, Xiaomin Chen, \IEEEmembership{Member, IEEE}, Li Wang, \IEEEmembership{Senior Member, IEEE},\\ Guangjie Han, \IEEEmembership{Fellow,~IEEE}, and Zhu Han, \IEEEmembership{Fellow,~IEEE}\\
    \thanks{
An earlier version of this paper was presented in part at the 2026  IEEE International Symposium on Personal, Indoor and Mobile Radio Communications (PIMRC)\cite{YWang2026-1}. \emph{(Corresponding author: Guangjie Han.)} \emph{(Ruoguang Li and Yu Wang are co-first authors.)}
		  
R. Li, Y. Wang, G. Han are with the College of Information Science and Engineering, Hohai University, Changzhou 213200, China. R. Li is also with the State Key Laboratory of ISN, Xidian University, Xi’an 710071, China  (e-mail: ruoguangli@hhu.edu.cn, wangyu227@hhu.edu.cn, hanguangjie@gmail.com).

B. Tian is with the China Mobile Research Institute,
Beijing 100053, China. (e-mail: tianbingxin@chinamobile.com).

B. Picano is with the Department of Information Engineering, University of Florence, 50139 Firenze, Italy. (e-mail: benedetta.picano@unifi.it).
          
Y. Chen is with the College of Information Science and Technology, Jinan University, Guangzhou 510632, China (e-mail: chenyy@jnu.edu.cn).

X. Chen is with the School of Information Science and Technology, Nantong University, Nantong 226019, China (email: chenxm@ntu.edu.cn).
        
L. Wang is with the School of Computer Science (National Pilot Software Engineering School), Beijing University of Posts and Telecommunications, Beijing 100876, China (e-mail: liwang@bupt.edu.cn).
        
Z. Han is with the Department of Electrical and Computer Engineering, University of Houston, Houston, TX 77004 USA (e-mail: hanzhu22@gmail.com).
	}
    \vspace{-5mm}
}


\maketitle

\begin{abstract}
Driven by the ever-growing demand for intelligent maritime services, cooperative unmanned aerial vehicle (UAV) and unmanned surface vehicle (USV) operations, together with wireless integrated sensing, communication, and computation (ISCC), have become pivotal in maritime environments. However, the distinctive maritime environment precludes a direct application of terrestrial
ISCC models, while sensing, communication, and computation compete for the same spatio-temporal resources, thus yielding inherent performance trade-offs that become particularly severe in
resource-constrained ocean scenarios. To address these issues, we investigate an ISCC system in which multiple UAVs detect targets at the sea surface and partially offload the collected data to a USV for cooperative processing. Characterizing the maritime-specific environment impairments, we incorporate sea clutter and wave-induced USV roll motion into the sensing and communication models, respectively. Moreover, we introduce data queueing models at both the UAVs and the USV to capture the dynamic coupling among sensing-data generation, wireless offloading, and computation execution. The optimization objective is to minimize the long-term average system energy consumption by jointly optimizing UAV–target association, UAV trajectories, beamforming, time allocation, offloading ratios, and computation resource allocation. Given that these decisions differ in update frequency and time sensitivity, we propose a Lyapunov-guided dual-timescale hierarchical soft actor-critic (LTH-SAC) framework, which transforms the long-term problem into a sequence of queue-aware online subproblems and coordinates optimization across the two timescales. Simulation results verify the effectiveness of the proposed framework and demonstrate its capability to balance the stability-performance trade-off among sensing, communication, and computation.
\end{abstract}

\begin{IEEEkeywords}
Integrated sensing, communication, and computation (ISCC), unmanned aerial vehicle (UAV), unmanned surface vehicle (USV), hierarchical reinforcement learning.                       
\end{IEEEkeywords}

\section{Introduction}
The rapid development of intelligent maritime services, such as ocean monitoring, ship navigation, and offshore inspection,
is driving an increasing demand for ubiquitous perception, reliable communication, and prompt data processing\cite{FSAlqurashi2023}. However, these functionalities are usually designed separately in conventional maritime networks, making it difficult to simultaneously guarantee relevant quality-of-service (QoS) in harsh maritime environments with limited resources and stringent energy budgets. The recent advancement of integrated sensing and communication (ISAC) \cite{FLiu2022,FLiu2020,SLiu2020} and multi-access edge computing (MEC)\cite{Mach2017Mobile,Mao2017MEC,Lin2024Maritime} have paved the way for a broader paradigm, namely integrated sensing, communication, and computation (ISCC), which establishes an end-to-end information processing framework from data acquisition and transmission to edge task execution \cite{Wen2025Survey,CLi2025,WChen2026}. 

Nevertheless, constrained by the fixed deployment, conventional terrestrial infrastructures can hardly cover vast open-sea areas, thereby struggling to adaptively provide high performance ISCC to the sparse and time-varying distribution of maritime users. In addition, constructing offshore facilities together with backhaul links inevitably causes prohibitive costs. Due to the high maneuverability,
on-demand deployment, and predominant line-of-sight (LoS) channels, unmanned aerial vehicles (UAVs) are envisioned as promising ISCC platforms that provide wide-area coverage along with flexible sensing,
communication, and computing capabilities for maritime
services~\cite{Meng2024UAVEnabled,Zhou2018Computation,Liu2023MLTrajectory}. In parallel, unmanned surface vehicles (USVs) are free from the stringent payload and endurance constraints of airborne platforms, and thus are equipped with relatively powerful computing resources.  Operating as maritime edge nodes, the USV supports information aggregation and on-site decision making, thereby serving as a complement to UAVs in accomplishing complex maritime ISCC tasks~\cite{ZWang2024,RZhang2026}.

Despite the above advantages, the design of the UAV-USV collaborative maritime ISCC system poses unique challenges. On the one hand, {\em the maritime environment invalidates several models adopted in terrestrial ISCC designs}. Specifically, maritime target sensing is impaired by sea clutter scattered from the irregular
sea surface, which changes the statistical characteristics of received echoes and degrades sensing performance~\cite{CarreteroMoya2010SeaClutter}. Moreover, the air-to-sea (A2S) channels in open-sea environments are jointly shaped by the LoS path and sea-surface reflection, which is significantly different from terrestrial links~\cite{Nomikos2023Survey}. Finally, the wave-induced vessel motion perturbs signal transmission and reception. In particular, the roll motion of USVs continuously rotates the onboard antenna array, varying the effective angle of arrival (AoA) and beam alignment~\cite{Soltani2011ShipAntenna}. On the other hand, {\em sensing, communication, and computation compete for the same
spatio-temporal resources, leading to inherent performance trade-offs}. Allocating more time to sensing improves information acquisition and estimation accuracy, yet generates more data to be transmitted and processed. Similarly, increasing computation offloading alleviates the local processing burden of UAVs, but imposes additional pressure on the A2S links and the computing resources of USVs. Consequently, persistent data backlogs may accumulate and eventually interrupt long-endurance missions, which is particularly severe in resource-constrained ocean scenarios. Therefore, for resource-constrained UAV-USV collaborative maritime ISCC systems, it is crucial to
jointly model the queue dynamics at both the UAV and the USV, so as to capture the interactions among sensing-data generation, communication offloading, and computation execution, thereby balancing system performance, energy consumption, and queue stability.
\subsection{Related Works}
Recent studies have investigated the integration of sensing, communication, and computation in UAV-enabled networks. For example, the authors in \cite{PHou2025} developed a distributed deep reinforcement learning framework for cooperative UAV-enabled ISCC systems, jointly optimizing bandwidth allocation, transmit power, and computing capacity under sensing and latency constraints. For multi-UAV-assisted ISCC systems, \cite{Deng2024MultiUAVISCC} jointly optimized user association, target assignment, deep neural network (DNN) splitting, beamforming, computation resource allocation, and UAV locations to meet latency and accuracy requirements. In \cite{Zhou2024UAVISCC}, UAV CPU frequency, sensing power, user transmit power, and trajectory were jointly optimized to minimize weighted energy consumption. Similarly, \cite{Chen2024UAVISCC} jointly optimized UAV trajectory, beamforming, and computation offloading to improve computing efficiency while satisfying sensing-quality requirements.
However, most existing UAV-assisted ISCC designs were developed for terrestrial or generic air-to-ground scenarios and do not properly capture maritime-specific propagation and platform dynamics. 
Unlike terrestrial environments, maritime environments exhibit unique propagation and sensing characteristics. In \cite{Wang2018Wireless}, the authors investigated maritime wireless channels and highlighted distinctive features such as sparse scattering and sea-wave-induced channel variations. In \cite{Zhang2025U2SChannel}, a geometry-based UAV-to-ship channel model was developed by explicitly incorporating line-of-sight propagation, sea-surface reflection, and ship oscillation, demonstrating the significant impact of maritime dynamics on channel non-stationarity. From the sensing perspective, \cite{NAhmed2025} incorporated sea clutter and wave shadowing into a non-stationary UAV-to-ship multiple-input multiple-output (MIMO) channel model to characterize their effects on received signals. Moreover, \cite{JYu2024} demonstrated that wave-induced rotational motion of USVs, particularly rolling, can cause severe beam misalignment and consequently degrade wireless transmission performance. 

Although maritime integrated wireless systems have attracted increasing attention, existing studies remain relatively limited and predominantly focus on ISAC. In \cite{Zhang2024MaritimeISAC}, the authors developed a shipboard maritime ISAC system and jointly optimized system deployment and transmission parameters to improve communication performance while satisfying sensing requirements. In \cite{Zhang2024MaritimeResource}, the authors investigated multi-user maritime ISAC resource allocation and revealed the trade-off between communication and sensing performance under different sensing requirements. For maritime emergency networks, \cite{Li2025MaritimeISAC} jointly optimized UAV trajectory, subcarrier assignment, and power allocation to maximize end-to-end communication rate while satisfying sensing QoS constraints.
At the broader network-architecture level, \cite{Dai2022SmartOceans} surveyed ISCC networks for smart oceans, discussing the collaboration of space, aerial, sea-surface, and underwater layers, while dedicated studies on joint optimization of maritime ISCC systems remain limited. The incorporation of computation further complicates resource coordination because sensing, communication, and computation compete for limited time, transmit power, bandwidth, and computing resources. \cite{Wen2025Survey} reviewed the interdependence among these functionalities and the resulting resource competition. Accordingly, several studies have investigated performance trade-offs in ISCC systems. \cite{Xu2023UAVISCC} characterized the trade-off between computation capacity and sensing beampattern gain in UAV-enabled ISCC systems, showing that improving one may compromise the other. Similarly, \cite{Liu2025ISCCOffloading} jointly optimized transmit beamforming and computation offloading in a three-tier ISCC framework to reduce task execution latency while satisfying sensing constraints.

To characterize dynamic performance trade-offs, queueing models and queue-aware resource management have been explored in MEC and ISAC systems.
In \cite{Zhang2019UAVMEC}, the authors modeled stochastic task arrivals using data queues and employed Lyapunov optimization to jointly optimize offloading, resource allocation, and UAV trajectory while ensuring queue stability. 
For maritime internet of things (IoT)-MEC systems, \cite{Bai2025Dynamic} developed a Lyapunov-based dynamic offloading scheme to balance energy consumption and delay while accounting for stochastic arrivals and queue stability. 
In \cite{Gao2026ISCCQueue}, the authors incorporated sensing-data offloading and queue stability into a Lyapunov drift-plus-penalty framework for online resource management. For maritime UAV-ISAC networks, \cite{Li2026MaritimeISACQueue} investigated queue-aware data collection by modeling buffered sensing observations at maritime buoys and coordinating sensing, communication, and UAV mobility. In underwater acoustic sensor networks, \cite{Gou2026UnderwaterISACQueue} incorporated AoI-related queue states into multi-agent ISAC optimization to balance information freshness and communication throughput. Despite these advances, the joint queue dynamics of sensing, communication, and computation functionalities in maritime ISCC systems remain under-explored. In particular, the coupling among sensing data generation, UAV-side offloading queues, and USV-side computation queues, as well as its impacts on long-term stability and energy-performance trade-offs under maritime-specific conditions, have not been adequately characterized.
\subsection{Main Contributions}
Inspired by the above discussions, this paper investigates a multi-UAV-assisted maritime ISCC network, where ISAC-enabled UAVs sense maritime targets, locally process sensing data, and offload remaining tasks to a USV for further computation and fusion. The objective is to minimize the long-term average energy consumption under sensing, computation, mobility, and queue-stability constraints. To solve the resulting coupled stochastic optimization problem, we develop a Lyapunov-guided dual-timescale hierarchical soft actor-critic (LTH-SAC) framework for joint long-term service planning and short-term resource allocation. The main contributions of this work are summarized as follows:
\begin{enumerate}
	\item We consider a maritime ISCC framework that captures the impacts of complex maritime environments and the coupling among sensing, communication, and computation. In particular, sea clutter is incorporated into the sensing model, while maritime channel characteristics and USV roll motion are considered in the communication model to characterize the dynamic variations of performance. Moreover, UAV-side and USV-side queues are introduced to capture the end-to-end data flow from sensing-data generation to wireless offloading and computation processing, revealing the stability-performance trade-off.
    
    \item We formulate a dual-timescale long-term stochastic optimization problem to balance system energy consumption, sensing performance, communication efficiency, and queue stability. By jointly optimizing UAV-target association, trajectory planning, beamforming, time allocation, computation offloading, and computing resource allocation, the proposed formulation characterizes the stability–performance trade-off. Considering the heterogeneous update frequencies and time sensitivities of different decisions, a dual-timescale framework is introduced, where large-timescale service planning and small-timescale resource adaptation are coordinated through Lyapunov drift-plus-penalty optimization.

    \item We propose a LTH-SAC algorithm to solve the resulting mixed-integer and highly coupled optimization problem. The high-level hybrid SAC policy determines long-term decisions including UAV-target association and trajectory planning, while the low-level multi-agent SAC (MASAC) policy performs adaptive resource allocation. Simulation results demonstrate that the proposed framework achieves superior energy efficiency and queue stability compared with benchmark algorithms, while revealing and effectively balancing the inherent trade-offs among sensing, communication, computation, and system stability.
\end{enumerate}

The remainder of this paper is organized as follows. Section \ref{sec:model} presents the considered multi-UAV-assisted maritime ISCC system model and problem formulation. Section \ref{sec:problem} applies Lyapunov optimization to derive a slot-wise drift-plus-penalty problem. Section \ref{sec:algorithm} introduces the proposed LTH-SAC algorithm. Section \ref{sec:results} presents the simulation results, and finally Section \ref{sec:conclusion} concludes the paper.

\section{System Model}
\label{sec:model}

As illustrated in Fig.~\ref{fig1}, we consider a multi-UAV-assisted maritime ISCC network, which consists of a USV, \(K\) maritime targets,
and $U$ UAVs equipped with a uniform linear array of $N_{u,t}$ transmit antennas and $N_{u,r}$ receive antennas, 
which are defined as the sets  $\mathcal{K} =\{1,2,\cdots,K\}$ and $\mathcal{U}=\{1,2,\cdots,U\}$. 
Each UAV is equipped with an ISAC device and a MEC server to perform radar sensing, communication and
computing. The UAVs first sense maritime targets and generate sensing data. Due to the limited onboard computation capability, the sensed data are partially processed locally by UAVs and partially offloaded to the USV for further computation and fusion. To characterize the different variation rates of long-term service decisions and delay-sensitive resource allocation, we introduce a two-timescale framework, as illustrated in Fig.~\ref{fig2}. Specifically, the flight duration of the UAVs is divided into $C$ large-timescale frames. The set of frame indices is denoted by $\mathcal{C} \equiv \{1, \cdots, C\}$, and we use $t_c$ to indicate the $c$-th frame. Each frame can be
divided into multiple small time slots, which are  indexed by $n \in \mathcal{N} \equiv \{1, \cdots, N\}$. The $n$-th small time slot in frame $c$ is denoted by $t_c^n$, and the duration of each slot is $\delta_t$. Accordingly, the duration of each frame is $T_c=N\delta_t$, and the total flight duration is $T=C N\delta_t$.
\begin{figure}[t]
	\centering
	\setlength{\abovecaptionskip}{0pt}
	\includegraphics[width=3.5in, trim=10pt 60pt 10pt 20pt, clip]{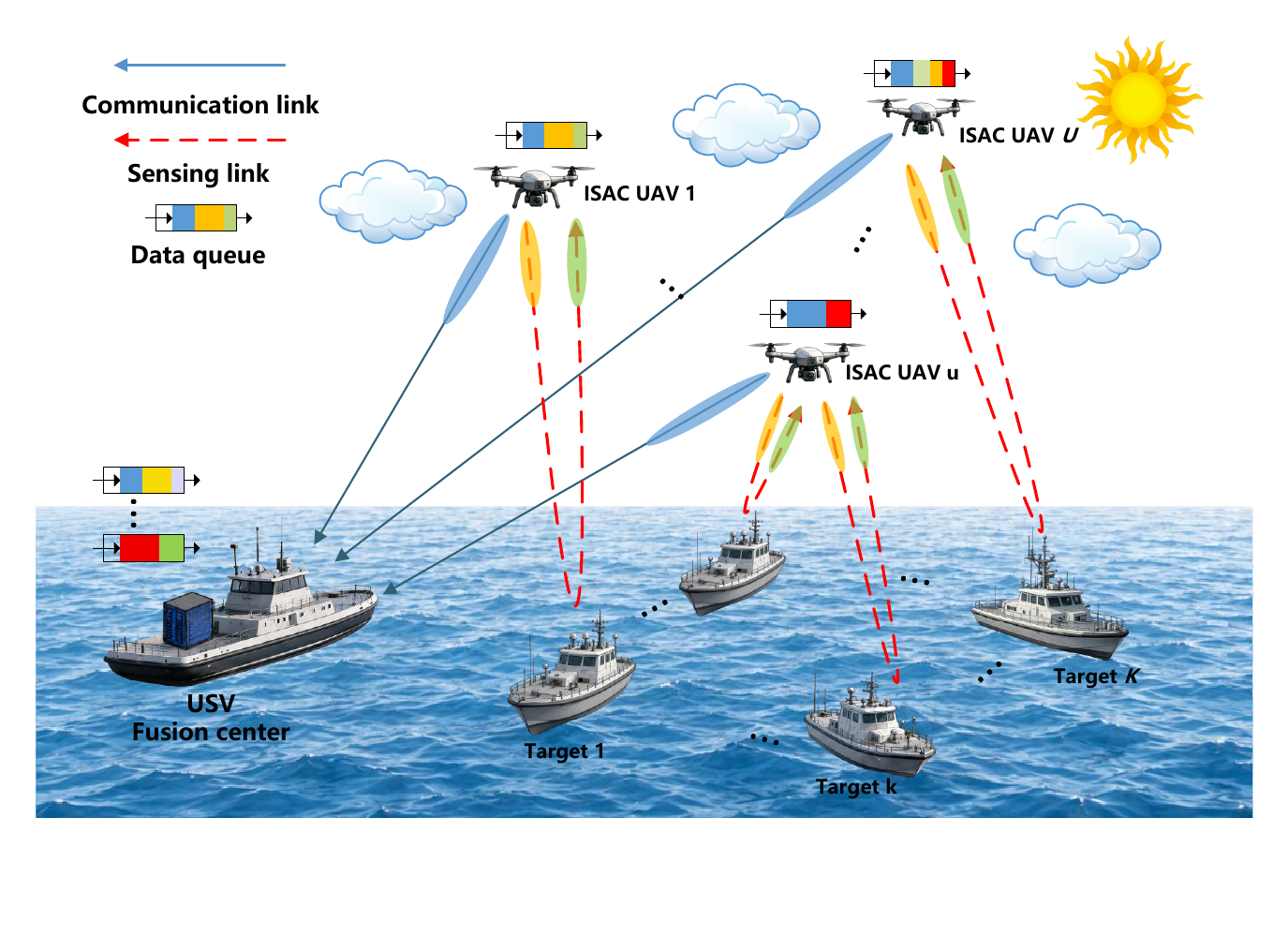}
	\caption{\small A multi-UAV-assisted maritime ISCC system model.}
	\label{fig1}
\end{figure}
According to the different update frequencies and decision impacts of the optimization variables, we classify them into large-timescale variables and small-timescale variables, as summarized in Table~\ref{tab:two_timescale_variables}. 
\begin{table}[t] 
	\centering
	\caption{Classification of Dual-Timescale Variables}
	\label{tab:two_timescale_variables}
	\renewcommand{\arraystretch}{1.15}
	\setlength{\tabcolsep}{3.5pt}
	\begin{tabular}{
			p{0.30\columnwidth} 
			>{\centering\arraybackslash}p{0.13\columnwidth} 
			p{0.51\columnwidth}
		}
		\hline
		\textbf{Variable} & \textbf{Timescale} & \textbf{Reason} \\
		\hline
		UAV-target association $\alpha_{u,k}[t_c]$
		& Large 
		& Frequent reassociation may introduce extra overhead and service instability. \\
		
		UAV trajectory $\mathbf{q}_u[t_{c}]$
		& Large 
		& 
		Frequent updates would introduce unnecessary motion control and propulsion costs. \\
		\hline
		Time allocation factor $\xi_u[t_c^n]$
		& Small 
		& Affects delay-sensitive sensing and data transmission, and depends on instantaneous channel adaptation. 
		\\
		
		Offloading ratio $\lambda_u[t_c^n]$
		& Small 
		& 
	    Needs fast adaptation to real-time channel quality, task backlog, and USV computing load. \\
		
		Sensing beamforming vector $\mathbf{p}_u[t_c^n]$
		& Small 
		& Depends on the instantaneous echo power and SCNR performance. \\
		
		Communication beamforming vector $\mathbf{f}_u[t_c^n]$
		& Small 
		& Highly sensitive to the instantaneous channel and USV roll motion. \\
		
		USV computing resource allocation $f_{u,0}[t_c^n]$
		& Small 
		& Directly affects real-time task execution delay and computation queues, and needs rapid adjustment according to the currently offloaded workloads. \\
		\hline
	\end{tabular}
\end{table}

We consider a three-dimensional (3D) Cartesian coordinate system, where the coordinates of the USV and target $k$ in frame $t_c$ are denoted by $\mathbf{q}_0[t_c]=(x_0[t_c],y_0[t_c],0)^T$ and $\mathbf{g}_k [ t_c ]=(x_k[t_c], y_k[t_c], 0)^{T}$. We consider that the waypoint of UAV $u$ at the beginning of frame $t_c$ is $\mathbf{q}_u [ t_c ]=(x_u [ t_c ], y_u[ t_c ],  H_u )^{T}$ with $H$ being the fixed altitude of each UAV, and the flight speed of UAV $u$ can be expressed as
$\mathbf v_u [t_c^n]
=
\frac{
	\mathbf q_u^H[t_{c+1}]-\mathbf q_u^H[t_c]
}
{N\delta_t}.$ 
In addition, we introduce a large-timescale binary variable $\alpha_{u, k} [ t_c ] \in\{0, 1 \} $ that indicates whether target $k$ is assigned to UAV $u$ during frame $t_c$. 
In order to improve the resource utilization efficiency, each target is assigned to at most one UAV \cite{PHou2025}, and each UAV only senses at most $K_u^{\rm max}$ targets due to hardware limitations, i.e., 
\begin{equation}
\sum_{u=1}^{U} \alpha_{u, k} [ t_c ] \leq1, \sum_{k=1}^{K} \alpha_{u, k} [ t_c ] \leq K_u^{\rm max}, \forall k \in \mathcal{K}, u \in \mathcal{U}, \forall t_c.
\end{equation}

Moreover, let $\lambda_u[t_c^n]\in[0,1]$ denote the offloading ratio of UAV $u$, while the remaining fraction $1-\lambda_u[t_c^n]$ is processed locally by UAV $u$.
To avoid interference, time division duplexing (TDD) is employed to perform radar sensing and communication. As illustrated in Fig.~\ref{fig2}, each time slot $t_c^n$ is further divided into two sub-slots. For UAV \( u \), the first sub-time slot $\xi_{u} [ t_c^n ] \delta_{t} $ is dedicated to target sensing, while the second sub-time slot $( 1-\xi_{u} [ t_c^n ] ) \delta_{t} $ is allocated for task offloading to the USV, where $\xi_{u} [ t_c^n ]\in{[0,1]}$ is the time slot allocation factor. After each sensing sub-time slot, local computation and task offloading are performed in parallel during the communication sub-slot. Furthermore, local computation in the current slot is assumed to be completed within the communication duration.
\subsection{Radar Sensing Model}
In the time slot $t_c^n$, each UAV first performs sensing to detect the targets on the sea surface. Let $s_{u,r}[t_c^n]$ denote the radar waveform emitted from the UAV $u$ with unit power, then the corresponding transmit signal can be expressed as
\begin{align}
\mathbf{x}_u[t_c^n]=\mathbf{p}_u[t_c^n]s_{u,r}[t_c^n],
\end{align}
where $\mathbf{p}_u[t_c^n]\in\mathbb{C}^{N_{u,t}\times 1}$ is the transmit beamforming vector.

Unlike terrestrial target detection, sea clutter is a non-negligible factor for maritime target detection. Specifically, sea clutter can be characterized as a composite Gaussian model from the perspective of physical scattering and mathematical tractability, which is represented as a spherically invariant random vector (SIRV), i.e., a product of two linearly independent components that reflect the roughness of the sea surface and can be expressed as\cite{Zhang2024MaritimeISAC}
\begin{equation}
	{\bf{c}}_{u} [t_c^n]=\sqrt{\tau} \mathbf{z}_u[t_c^n],
\end{equation}
where $\tau$ denotes the texture component, which generally remains constant in a coherent processing interval (CPI) and represents the macroscopic structure formed by long gravity waves. $\mathbf{z}_u[t_c^n]\sim\mathcal{CN}(\mathbf{0},\sigma^{2}_{u,z} [t_c^n] \mathbf{I}_{N_{u,r}})$ denotes the speckle component that represents the microscopic structure formed by capillary waves.
\begin{figure}[t]
	\includegraphics[width=0.5\textwidth, trim=30pt 18pt 10pt 50pt, clip]{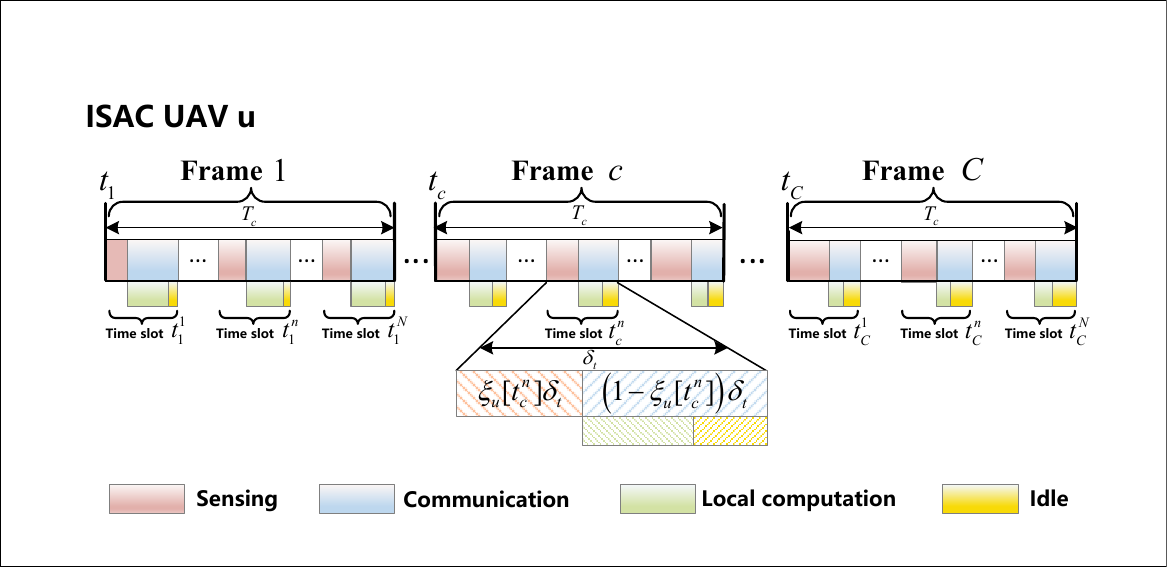}
	\centering
	\caption{\small Frame and slot structure of the maritime ISCC system.}\label{fig2}
	\vspace{-1mm}
\end{figure}
Then, the sensing echo received by UAV $u$ in time slot $t_c^n$ is
\begin{align}
&\mathbf{r}_{u} [ t_c^n ]=\sum_{k=1}^{K} \alpha_{u, k} [ t_c ] \beta_{u,k}[t_c^n] \mathbf{a} ( \theta_{u,k} [ t_c^n ] ) \mathbf{b}^{H} ( \theta_{u,k} [ t_c^n ] ) \mathbf{x}_u[t_c^n]\nonumber\\
&+\sum^{U}_{i=1,i\neq u}\sum^{K}_{l=1}\alpha_{i,l}[t_c]\beta_{i,l}[t_c^n] \mathbf{a}( \theta_{u,l} [ t_c^n ] ) \mathbf{b}^{H} ( \theta_{i,l} [ t_c^n ] ) \mathbf{x}_i[t_c^n]\nonumber\\
&+{\bf{c}}_{u} [t_c^n]+\mathbf{n}_{u,r} [ t_c^n ],
\end{align}
where $\mathbf{n}_{u,r} [t_c^n]\sim\mathcal{CN}(\mathbf{0},\sigma^2_{u,r}[t_c^n]\mathbf{I}_{N_{u,r}})$ represents the additive white Gaussian noise (AWGN). 
$\beta_{u,k}[t_c^n]$ is the reflection coefficient of target $k$ for UAV $u$, which is modeled as
$
	\beta_{u,k}[t_c^n]
	=
	\frac{
		\sqrt{
			\bar{\sigma}_{k}(f_c,\mathbf{Q}_k)
			B^{\rm ang}_{u,k}[t_c^n]
			B^{\rm fluc}_{k}[t_c^n]
		}
	}{2d_{u,k}[t_c^n]}
	e^{j\psi_{u,k}[t_c^n]},
$
where $\bar{\sigma}_{k}(f_c,\mathbf{Q}_k)$ denotes the nominal radar cross section (RCS), which depends on carrier frequency $f_c$ and target geometry parameter $\mathbf{Q}_k$. 
$d_{u,k}[t_c^n]=\|\mathbf{q}_u[t_c^n]-\mathbf{g}_k[t_c]\|$ denotes the distance between UAV $u$ and target $k$ in time slot $t_c^n$. Other parameters are defined in \cite{chen2009mimo}.
$\mathbf{a}(\theta_{u,k}[t_c^n])$ and $\mathbf{b}(\theta_{u,k}[t_c^n])$ are the receive and transmit steering vectors, respectively, and can be expressed as
\begin{equation}
	\resizebox{1\hsize}{!}{$\begin{aligned}
			{\bf{a}}( \theta_{u,k} [ t_c^n ])=[ 1, e^{j\pi \operatorname{s i n} ( \theta_{u,k} [ t_c^n ])},\cdots, e^{j\pi \operatorname{s i n} ( \theta_{u,k} [ t_c^n ]) ( N_{u,r}-1 )} ]^{T},\\
			{\bf{b}}( \theta_{u,k} [ t_c^n ])=[ 1, e^{j\pi \operatorname{s i n} ( \theta_{u,k} [ t_c^n ])},\cdots, e^{j\pi\operatorname{s i n} ( \theta_{u,k} [ t_c^n ]) ( N_{u,t}-1 )} ]^{T},
		\end{aligned}$}
\end{equation}
where $\theta_{u,k}[t_c^n]=\arcsin\!\left(\frac{H}{\|\mathbf{q}_u[t_c^n]-\mathbf{g}_k[t_c]\|}\right)$ is the elevation angle from the target $k$ to UAV $u$. Let $\mathbf{A}(\theta_{u,k} [t_c^n])\triangleq\mathbf{a} ( \theta_{u,k} [ t_c^n ] ) \mathbf{b}^{H} ( \theta_{u,k} [ t_c^n ] )$ and $\mathbf{B}( \theta_{u,l} [ t_c^n ],\theta_{i,l} [ t_c^n ] )\triangleq\mathbf{a} ( \theta_{u,l} [ t_c^n ] ) \mathbf{b}^{H} ( \theta_{i,l} [ t_c^n ] )$.
$\mathbf{r}_{u}[t_c^n]\in\mathbb{C}^{N_{u,r}\times 1}$ is filtered by a set of receive beamforming vectors $\{\mathbf{w}_{u,k}[t_c^n]\}$, and the corresponding output is given as
\begin{align}
&r_{u,k}[t_c^n]=\mathbf{w}_{u,k}^{H}[t_c^n]\mathbf{r}_{u}[t_c^n]\nonumber\\
&=\mathbf{w}_{u,k}^{H}[t_c^n]\sum_{k=1}^{K} \alpha_{u, k} [ t_c ] \beta_{u,k}[t_c^n] \mathbf{A} ( \theta_{u,k} [ t_c^n ] )\mathbf{x}_u[t_c^n]\nonumber\\
&+\mathbf{w}_{u,k}^{H}[t_c^n]\sum^{U}_{i=1,i\neq u}\sum^{K}_{l=1}\alpha_{i,l}[t_c]\beta_{i,l}[t_c^n] \mathbf{B}( \theta_{u,l} [ t_c^n ],\theta_{i,l} [ t_c^n ] ) \mathbf{x}_i[t_c^n]\nonumber\\
&+\mathbf{w}_{u,k}^{H}[t_c^n]{\bf{c}}_{u} [t_c^n]+\mathbf{w}_{u,k}^{H}[t_c^n]\mathbf{n}_{u,r} [ t_c^n ].
\end{align}
Then, the received SCNR at the UAV $u$ for the target $k$ at time slot $t_c^n$ can be expressed as
\begin{align}
\label{SCNR}
\gamma_{u,k}[t_c^n]
=\frac{\alpha_{u,k}[t_c]|\beta_{u,k}[t_c^n]|^2\left|\mathbf{w}_{u,k}^{H}[t_c^n]\mathbf d_{u,k}[t_c^n]\right|^2}{\mathbf{w}_{u,k}^{H}[t_c^n]\mathbf{\Psi}_{u,k}[t_c^n]\mathbf{w}_{u,k}[t_c^n]},
\end{align}
where
\begin{align}
	\mathbf{\Psi}_{u,k}[t_c^n]
	=&\mathbf{R}_1[t_c^n]+\mathbf{R}_2[t_c^n]+(\sigma_{u,z}^2[t_c^n]+\sigma^2_{u,r}[t_c^n])\mathbf{I}_{N_{u,r}},\\
	\mathbf{R}_1[t_c^n]
	=& \sum_{\tilde{k}=1,\tilde{k}\neq k}^{K}
	\alpha_{u,\tilde{k}}[t_c]|\beta_{u,\tilde{k}}[t_c^n]|^2
	\mathbf{d}_{u,\tilde{k}}[t_c^n]
	\mathbf{d}_{u,\tilde{k}}^{H}[t_c^n],\\
	\mathbf{R}_2[t_c^n]
	=& \sum_{i=1,i\neq u}^{U}\sum_{l=1}^{K}
	\alpha_{i,l}[t_c]|\beta_{i,l}[t_c^n]|^2
	\mathbf{e}_{u,i,l}[t_c^n]
	\mathbf{e}_{u,i,l}^{H}[t_c^n],
\end{align}
and $\mathbf d_{u,k}[t_c^n]
\triangleq
\mathbf A(\theta_{u,k}[t_c^n])\mathbf p_u[t_c^n]$, $\mathbf{d}_{u,\tilde{k}}[t_c^n]
\triangleq
\mathbf{A}(\theta_{u,\tilde{k}}[t_c^n])\mathbf{p}_{u}[t_c^n]$, $\mathbf{e}_{u,i,l}[t_c^n]
\triangleq
\mathbf{B}(\theta_{u,l}[t_c^n],\theta_{i,l}[t_c^n])
\mathbf{p}_{i}[t_c^n]$.\\
For given transmit beamforming vectors, UAV trajectories and target association variables, the optimal receive beamforming vector $\mathbf w_{u,k}^{\star}[t^n_c]$ can be obtained by solving the equivalent minimum variance distortionless response
(MVDR) problem for SCNR maximization, which is given by \cite{Lorenz2005MVDR}
\begin{equation}
	\label{receiving beamforming}
	\mathbf w_{u,k}^{\star}[t_c^n]
	=
	\frac{
		\mathbf{\Psi}_{u,k}^{-1}[t_c^n]\mathbf d_{u,k}[t_c^n]
	}{
		\mathbf d^H_{u,k}[t_c^n]\mathbf{\Psi}_{u,k}^{-1}[t_c^n]\mathbf d_{u,k}[t_c^n]
	}.
\end{equation}
By substituting \eqref{receiving beamforming} into \eqref{SCNR}, the SCNR can be re-expressed as
\begin{equation}
\gamma_{u,k}^{\star}[t_c^n]
=
\alpha_{u,k}[t_c]
|\beta_{u,k}[t_c^n]|^2
\mathbf d_{u,k}^H[t_c^n]
\mathbf{\Psi}_{u,k}^{-1}[t_c^n]
\mathbf d_{u,k}[t_c^n].
\end{equation}
Therefore, the radar estimation information rate at the UAV $u$ echoed by target $k$ can be expressed as\cite{SPeng2025}
\begin{align}
R^{\rm sen}_{u,k}[t_c^n]=\frac{\delta}{2\mu}\log_2(1+\gamma_{u,k}^{\star}[t_c^n]),
\end{align}
where $\delta$ and $\mu$ denote the radar duty factor and pulse duration of the radar, respectively.
The corresponding amount of sensing data is
\begin{align}
C^{\rm sen}_{u,k}[t_c^n]=\xi_u[t_c^n]\delta_tR^{\rm sen}_{u,k}[t_c^n].
\end{align}
Then, the total sensed data of UAV $u$ in the slot $t_c^n$ is given by
\begin{equation}
	C_u^{\mathrm{sen}}[t_c^n]=\sum_{k=1}^{K}\alpha_{u,k}[t_c]C_{u,k}^{\mathrm{sen}}[t_c^n].
\end{equation}
Accordingly, the sensing data processed locally by UAV $u$ at time slot $t_c^n$ is $(1-\lambda_u[t_c^n])C^{\rm sen}_u[t_c^n]$, and the data required to be offloaded to the USV is $\lambda_u[t_c^n]C^{\rm sen}_u[t_c^n]$.

\subsection{Communication Model}
Once the UAV finishes sensing task, it will perform local computation and offload part of the sensing data to the USV for further processing simultaneously. It is worth noting that the communication between the UAV and USV is significantly influenced by both the maritime environment and the motion state of the USV. Additionally, since there are few obstructions on the sea surface, the transmission link between the UAV and USV is generally LoS-dominated. Therefore, the channel coefficient vector between the UAV $u$ and the USV can be written as\cite{BLi2025}
\begin{align}
\label{comm_ch}
\mathbf h_u[t_c^n]
=
\sqrt{\rho_u[t_c^n]}
\mathbf b_{\rm eff}
(
\widetilde{\theta}_u[t_c^n]),
\end{align}
where $\rho_u[t_c^n]$ accounts for the large-scale fading attenuation with respect to path loss. $\mathbf b_{\rm eff}
(\widetilde{\theta}_u[t_c^n])
=
\left[
1,
e^{j\pi\sin\widetilde{\theta}_u[t_c^n]},
\cdots,
e^{j\pi(N_{u,t}-1)\sin\widetilde{\theta}_u[t_c^n]}
\right]^T$ is the roll-aware equivalent steering vector with $\widetilde{\theta}_u[t_c^n]$ being the coordinate-rotation-based equivalent elevation angle from the USV to the UAV.

\subsubsection{Model of $\rho_u[t_c^n]$}
In practice, the LoS link and the surface reflection link are two primarily considered paths in an A2S channel. Therefore, the curved-earth two-ray (CE2R) model is adopted to formulate the path loss\cite{Wang2018Wireless, CZeng2021}, which is written by
\begin{equation}
\resizebox{1\hsize}{!}{$\begin{aligned}
	L_{\rm CE2R}[t_c^n]
	=
	-10\log_{10}
	\left(
	\left(\frac{\lambda}{4\pi d_u[t_c^n]}\right)^2
	\left(
	2\sin\left(
	\frac{2\pi h_t h_r}{\lambda d_u[t_c^n]}
	\right)
	\right)^2
	\right),
\end{aligned}$}
\end{equation}
where $d_u[t_c^n]=\sqrt{(x_u[t_c^n] - x_0[t_c])^2 + (y_u[t_c^n] - y_0[t_c])^2}$ denotes the horizontal distance between UAV $u$ and the USV in time slot $t_c^n$. $\lambda$ is the wavelength. $h_t$ and $h_r$ represent the heights of the transmit and receive antennas, respectively. 

Therefore, the large-scale fading attenuation $\rho_u[t_c^n]$ in \eqref{comm_ch} can be expressed as
\begin{align}
\rho_u[t_c^n]=10^{\frac{-L_{{\rm CE2R}}[t_c^n]} {1 0}}.
\end{align}
\begin{figure}[t]
	\includegraphics[width=0.5\textwidth, trim=40pt 50pt 20pt 160pt, clip]{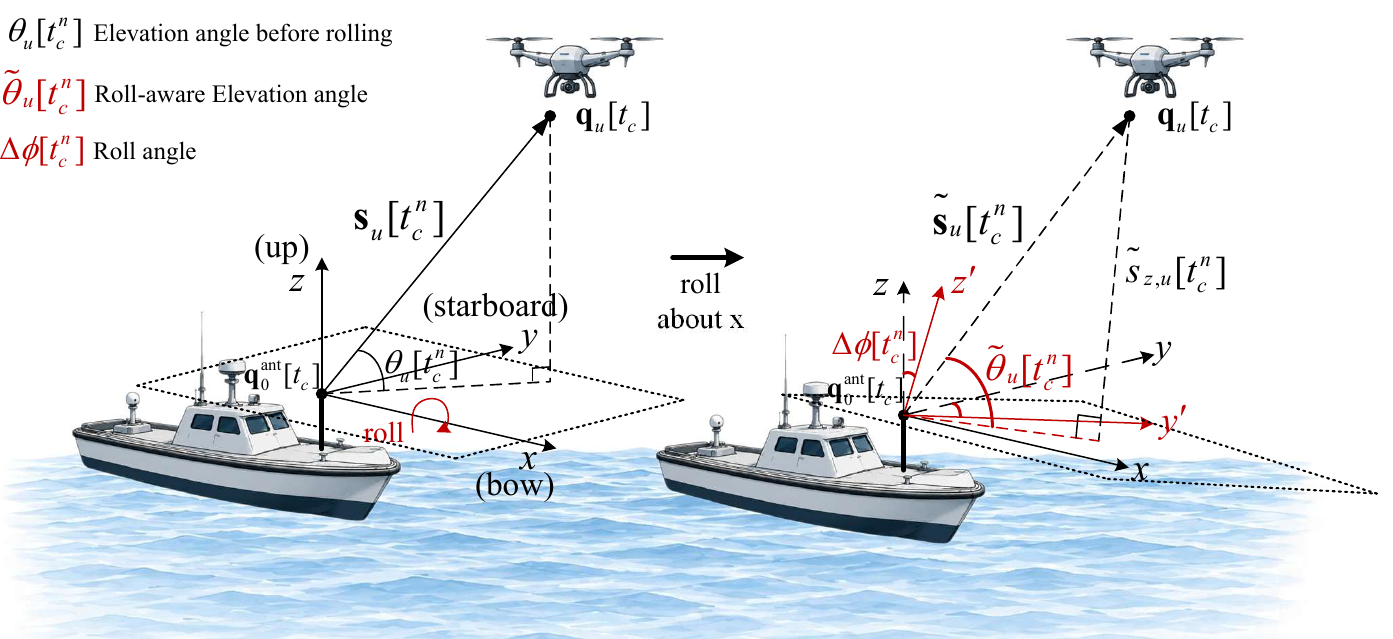}
	\centering
	\caption{\small Coordinate-rotation-based modeling of the roll-aware incident direction for the UAV–USV communication link.}\label{fig3}
	\vspace{-1mm}
\end{figure}
\subsubsection{Model of $\widetilde{\theta}_u[t_c^n]$}
As illustrated in Fig.~\ref{fig3}, we introduce a body-fixed coordinate system to describe the roll motion of the USV\cite{NAhmed2025}, where the $x$-axis points to the bow, the $y$-axis points to the starboard, and the $z$-axis points vertically upward respectively. The roll motion of the USV changes the local coordinate system of the receiving antenna into the coordinate system with axes $x$, $y'$ and $z'$, which further affects the equivalent incident direction. For simplicity,
let $\mathbf q_0^{\rm ant}[t_c]=(x_0[t_c],y_0[t_c],h_r)^T$ denote the USV antenna phase center. The unit link direction vector from the USV to UAV $u$ can be written as
\begin{equation}
\mathbf s_u[t_c^n]
=
\frac{
	\mathbf q_u[t_c^n]
	-
	\mathbf q_0^{\rm ant}[t_c]
}{
	\left\|
	\mathbf q_u[t_c^n]
	-
	\mathbf q_0^{\rm ant}[t_c]
	\right\|
}.
\end{equation}
Then, the roll rotation matrix corresponding to roll angle $\Delta\phi[t_c^n]$ is given by
\begin{equation}
\mathbf R_x(\Delta\phi[t_c^n])
=
\begin{bmatrix}
	1 & 0 & 0\\
	0 & \cos\Delta\phi[t_c^n] & -\sin\Delta\phi[t_c^n]\\
	0 & \sin\Delta\phi[t_c^n] & \cos\Delta\phi[t_c^n]
\end{bmatrix}.
\end{equation}

To characterize the incident direction observed in the roll-rotated local coordinate system of the USV, the global link direction vector $\mathbf s_u[t_c^n]$ is transformed into the USV local coordinate system as
\begin{equation}
\widetilde{\mathbf s}_u[t_c^n]
=
\mathbf R_x^T(\Delta\phi[t_c^n])
\mathbf s_u[t_c^n]=\left(
\widetilde{s}_{x,u}[t_c^n],\;
\widetilde{s}_{y,u}[t_c^n],\;
\widetilde{s}_{z,u}[t_c^n]
\right)^T,
\end{equation}
where $\widetilde{s}_{x,u}[t_c^n]$, $\widetilde{s}_{y,u}[t_c^n]$ and $\widetilde{s}_{z,u}[t_c^n]$ denote the $x$-axis, $y$-axis and $z$-axis components of $\widetilde{\mathbf s}_u[t_c^n]$, respectively, in the roll-rotated USV local coordinate system.

Then, the roll-aware equivalent elevation angle is obtained as
\begin{equation}
\widetilde{\theta}_u[t_c^n]=
\arcsin
\left(
\widetilde{s}_{z,u}[t_c^n]
\right).
\end{equation}
\subsubsection{Model of $\Delta\phi[t_c^n]$}
Under irregular wave conditions, the USV experiences variations in height and position due to waves. The rolling moment induced by wave disturbances is primarily related to the effective wave inclination angle $\phi_{e} (t)$, which can be modeled as\cite{JYu2024}
\begin{align}
	\phi_e(t) =
	K_1 K_2 \sin \chi \sum_{i=1}^{M} \left[ \sqrt{2S(\omega_i)\Delta\omega_i} \cos(\omega_{i}t + \varepsilon_i) \right],
\end{align}
where $M$ represents the number of waves in a random wave spectrum, and \(S(\omega)\) denotes the spectral density function of wave inclination angles. \(\omega_i\) is the angular frequency of the \(i\)-th wave,
$\Delta\omega_{i}$ is the frequency increment, and $\varepsilon_{i}$ represents the initial phase uniformly distributed within the range $[ 0, 2 \pi) $. $K_1$ and
$K_2$ denote the correction coefficients for finite ship width and
finite draft depth, respectively, $\chi$ is the angle of arrival.

Using the dynamic equations and the relationship between the wave disturbance moments and the effective wave inclination angle, the corresponding Laplace transform can be performed to derive the transfer function between the rolling angle and the effective wave inclination angle as
\begin{equation}
	G \left( s \right) \!=\!\frac{\Delta\phi\left( s \right)} {\phi_{e} \left( s \right)} \!=\! \frac{\omega_{\Delta\phi}^{2} \left( \frac{\Delta J_{\Delta\phi}} {D h} s^{2}+\frac{2 \zeta_{\Delta\phi}} {\omega_{\Delta\phi}} s \!+\! 1 \right)} {s^{2}+2 \zeta_{\Delta\phi} \omega_{\Delta\phi} s+\omega_{\Delta\phi}^{2}} ,
\end{equation}
where the parameters follow the standard ship roll model \cite{JYu2024}.

Therefore, the continuous-time roll response is obtained as
$
\Delta\phi_c(\tau)
= \sum_{i=1}^{M}A_i
\left|G(j\omega_{i})\right|
\cos
\left(
\omega_{i}\tau
+
\varepsilon_i
+
\angle G(j\omega_{i})
\right),
$
where $A_i=
K_1K_2\sin\chi
\sum_{i=1}^{M}
\sqrt{2S(\omega_i)\Delta\omega_i}$.

Then, the roll angle used in the slot $t_c^n$ is obtained by sampling the continuous-time roll response at the slot, i.e.,
\begin{equation}
\Delta\phi[t_c^n]
=
\Delta\phi_c([(c-1)N+n-1]\delta_t).\\
\end{equation}

Let $s_{u,c} [ t_c^n ] \sim {\cal C N} ( 0, 1 ) $ denote the signal transmitted by UAV $u$ to USV. By applying the beamforming vector $\mathbf{f}_u[t_c^n]$, the transmitted communication signal at the UAV $u$ in the time slot $t_c^n$ is given by
\begin{align}
\mathbf{x}_{u,c}[t_c^n]=\mathbf{f}_{u}[t_c^n]s_{u,c}[t_c^n].
\end{align}

Assume that the dedicated sensing signal is prior known by the USV, such that the sensing interference can be cancelled in advance. Therefore, the received signal at the USV in the time slot $t_c^n$ is represented as
\begin{equation}
	\mathrm{y}[ t_c^n ]=\sum_{u=1}^{U} {\bf{h}}_{u}^{H} [ t_c^n ] {\bf{f}}_{u} [ t_c^n ] s_{u,c} [ t_c^n ]+n_{u ,c} [ t_c^n ],
\end{equation}
where $n_{u,c} [t_c^n]\sim\mathcal{CN}( 0, \sigma_{u,c}^{2} )$ represents the AWGN. And the achievable rate from UAV $u$ to the USV at time slot $t_c^n$ can be written as
\begin{equation}
    R^{\rm com}_{u} [ t_c^n ]=\log_2 \left( 1+\frac{\left| {\bf{h}}_{u}^{H} [ t_c^n ] {\bf{f}}_{u} [ t_c^n ] \right|^{2}} {\sum_{i=1, i \neq u}^{U} \left| {\bf{h}}_{i}^{H} [ t_c^n ] {\bf{f}}_{i} [ t_c^n ] \right|^{2}+\sigma_{u,c}^{2}} \right).
\end{equation}

Therefore, the maximum amount of the offloading data at UAV $u$ can be expressed as
\begin{align}
C^{\rm com}_u[t_c^n]=B_u(1-\xi_u[t_c^n])\delta_tR^{\rm com}_u[t_c^n].
\end{align}
\subsection{Computation Model}
Each UAV is equipped with a MEC server to process tasks locally. Let $f_{u,\mathrm{loc}}[t_c^n]$ denote the local CPU-cycle frequency (cycle/s) of UAV $u$ in slot $t_c^n$, and let $l_u[t_c^n]$ denote the CPU cycles (in cycles/bit) required for processing one bit of local data. Then, the local computation rate of UAV $u$ can be expressed as
\begin{equation}
	R_{u,\text{loc}}^{\text{comp}}[t_c^n] = \frac{f_{u,\text{loc}}[t_c^n]}{l_{u}[t_c^n]}.
\end{equation}

Accordingly, the maximum amount of data that can be locally processed by UAV $u$ in slot $t_c^n$ is
\begin{equation}
	C_{u,\mathrm{loc}}^{\mathrm{comp}}[t_c^n]=R_{u,\mathrm{loc}}^{\mathrm{comp}}[t_c^n](1-\xi_{u}[t_c^n])\delta_t.
	\end{equation}
	
For simplicity, we consider the locally processed tasks are required to be completed within the current slot. Therefore, the amount of data assigned to local computing cannot exceed the local computation capability of UAVs:
\begin{equation}
	(1-\lambda_u[t_c^n])C_u^{\mathrm{sen}}[t_c^n]\le C_{u,\mathrm{loc}}^{\mathrm{comp}}[t_c^n], \quad \forall u, c, n.
\end{equation}

In addition, let $l_0 [t_c^n]$ denote the required CPU cycles for processing one bit of data, and $f_{u,0} [t_c^n]$ represent the CPU frequency allocated by the USV to UAV $u$. Then the computation rate at the USV is given by
\begin{equation}
	R^{\rm comp}_{u,0} [t_c^n]=\frac{f_{u,0} [t_c^n]} {l_0 [t_c^n]}.
\end{equation}

Then, the number of computation bits computed by the USV for UAV $u$ at time slot $t_c^n$ can be given as
\begin{equation}
	C_{u,0}^{\rm comp}[t_c^n]=R^{\rm comp}_{u,0} [t_c^n] \delta_t.
\end{equation}

To characterize the dependence of classification performance on
the amount of sensing data, we adopt a sensing-data-based
classification-error surrogate, 
which is given by \cite{Zhang2022Accelerating}
\begin{align}
e_u[t_c^n]=a_u[t_c^n](\max\left\{
N_u^{\rm sam}[t_c^n],\,1
\right\})^{-b_u[t_c^n]},
\end{align}
where $a_u[t_c^n]$ and $b_u[t_c^n]$ are hyperparameters of difficulty of the data from the UAV $u$. $N_u^{\rm sam}[t_c^n]$ is the number of samples that the communication can afford, which can be written as
$
N_u^{\rm sam}[t_c^n]=\frac{C^{\rm sen}_u[t_c^n]}{D_s},
$
where $D_s$ denotes the data size per sample.
\subsection{Energy Consumption Model}
\subsubsection{USV Energy Consumption Model}
The USV energy consumption at time slot $t_c^n$ is given by
\begin{align}
E_{0} [ t_c^n ]=& \sum_{u=1}^{U} \kappa_0 \delta_t \left( f_{u,0} [ t_c^n ] \right)^{3}, 
\end{align}
where $\kappa_0$ is a constant related to the hardware specifications of the USV.
\subsubsection{UAV Energy Consumption Model}
The UAV energy consumption consists of propulsion, computation, sensing, and offloading energy. Specifically, the propulsion energy of UAV $u$ at time slot $t_c^n$ is given by \cite{Zeng2019UAVEnergy}
\begin{align}
		 E^{\rm prop}_{u}[t_c^n]=& \left(P_{\rm bla} \Biggl( 1+\frac{3 \| {\bf{v}}_{u} [ t_c^n ] \|^{2}} {U_{\mathrm{t i p}}^{2}} \Biggr)+\frac{1} {2} d_{0} \rho_{0} s A \| {\bf{v}}_{u} [ t_c^n ] \|^{3}\right.\nonumber \\ 
         &\left.+ \, P_{\rm ind} \! \left( \sqrt{1+\frac{\| {\bf{v}}_{u} [ t_c^n ] \|^{4}} {4 v_{0}^{4}}}-\frac{\| {\bf{v}}_{u} [ t_c^n ] \|^{2}} {2 v_{0}^{2}} \right)^{\frac{1} {2}}\right)\delta_t, 
\end{align}
where the parameters follow the conventional rotary-wing UAV energy model \cite{Zeng2019UAVEnergy}.

The local computation energy consumption of UAV $u$ in the slot $t_c^n$ is modeled as
\begin{equation}
	E_{u}^{\mathrm{loc}}[t_c^n]=\kappa_u (1-\lambda_u[t_c^n])C^{\rm sen}_u[t_c^n]l_u[t_c^n] \left(f_{u,\mathrm{loc}}[t_c^n]\right)^2,
\end{equation}
where $\kappa_u$ is the effective switched capacitance coefficient related to the hardware capability of UAV $u$.

In addition, the sensing and offloading energy consumption  of the UAV $u$ are respectively given by
\begin{equation}
	E^{\rm sen}_u[t_c^n]=\|\mathbf{p}_u[t_c^n]\|^2\xi_u[t_c^n]\delta_t,
\end{equation}
\begin{equation}
	E_{u}^{\mathrm{o f f}} [ t_c^n ]=\left\| {\bf{f}}_{u} [ t_c^n ] \right\|^{2} \left( 1-\xi_{u} [ t_c^n ] \right) \delta_{t}.
\end{equation}

Therefore, the total energy consumption of UAV $u$ is
$
	E_{u}[t_c^n]= E^{\rm prop}_u[t_c^n]+E^{\rm loc}_u[t_c^n]+E^{\rm sen}_u[t_c^n]+E^{\rm off}_u[t_c^n],
$
and the total energy consumption of the system is
$
	E_{sum}[t_c^n]= \sum_{u=1}^{U}E_u[t_c^n]+\eta E_0[t_c^n],
$
where \(\eta\) is a weighting coefficient for the USV computation energy consumption.

\subsection{Dynamic Queue Scheduling Model}
It is worth noting that sensing and computation offloading generate continuous data flows at both UAV and USV sides. At the UAV side, increasing sensing time improves inference performance but produces more sensing data, while allocating more communication time enhances offloading capability, leading to a sensing–communication trade-off. When the generated data exceed the transmission capability, UAV-side data backlog occurs and may degrade network stability. Similarly, at the USV side, the aggregated offloading tasks from multiple UAVs may exceed the computation service capability, resulting in computation backlog. Therefore, UAV-side and USV-side queues are constructed to characterize the end-to-end data flow and maintain system stability.
For notational convenience, we define the next small slot of \(t_c^n\) as
\begin {equation}
\operatorname {next}(t_c^n)=
\begin {cases}
t_c^{n+1}, & n<N,\\
t_{c+1}^{1}, & n=N.
\end {cases}
\end {equation}

Specifically, in the time slot $t_c^n$, we define the amount of data actually transmitted from UAV $u$ to the USV as \(D_u^{\text{off}}[t_c^n]=\min\left\{Q_u^A[t_c^n]+\lambda_u[t_c^n] C_u^{\text{sen}}[t_c^n],C_u^{\text{com}}[t_c^n]\right\}\), then the data offloading queue at the UAV $u$ can be expressed as $Q_u^A[\operatorname{next}(t_c^n)]=\max\{Q_u^A[t_c^n]-C_u^{\rm com}[t_c^n]+\lambda_u[t_c^n]C_u^{\rm sen}[t_c^n],0\}$. Similarly, the computation queue corresponding to the offloaded tasks from UAV $u$ at the USV is updated as
	$Q_u^V[\operatorname{next}(t_c^n)]
	=
	\max\{Q_u^V[t_c^n]-C_{u,0}^{\rm comp}[t_c^n],0\}
	+
	D_u^{\text{off}}[t_c^n]$.

To ensure the stability of the data queue, the queue length of all devices should not increase infinitely, i.e.,
\begin{equation}
	\lim_{C\to+\infty}
	\frac{1}{CN}
	\sum_{c=1}^{C}
	\sum_{n=1}^{N}
	\mathbb E(Q_u^A[t_c^n])
	< \infty,\quad \forall u.
\end{equation}
\begin{equation}
	\lim_{C\to+\infty}
	\frac{1}{CN}
	\sum_{c=1}^{C}
	\sum_{n=1}^{N}
	\mathbb E(Q_u^V[t_c^n])
	< \infty,\quad \forall u.
\end{equation}
\subsection{Problem Formulation}
We aim to minimize long-term time-averaged total energy consumption subject to the system stability and performance requirement. Under the two-timescale system model, let $\Omega=\{\{\alpha_{u,k}[t_c]\},\{\mathbf{q}_u[t_c]\},\{\mathbf{p}_u[t_c^n]\},\{\mathbf{f}_u[t_c^n]\},\{\xi_u[t_c^n]\},\{\lambda_u[t_c^n]\},\\\{f_{u,0}[t_c^n]\}\}$, then the optimization problem can be formulated as
\begin{subequations}
	\begin{align}
		\text{(P1):}\quad &
		\min_{\Omega}\quad
		\lim_{C\to+\infty}
		\frac{1}{CN}
		\sum_{c=1}^{C}
		\sum_{n=1}^{N}
		\mathbb E\left(E_{\rm sum}[t_c^n]\right)\\
		\mathrm{s.t.}\quad
		&\xi_{u}[t_c^n]R_{u,k}^{\rm sen}[t_c^n]
		\ge
		\alpha_{u,k}[t_c]R_{\rm th}^{\rm sen},
		\quad \forall u,k,c,n, \label{eq:sensing}\\
		&\|\mathbf p_u[t_c^n]\|^2\le P_u^{\max},
		\quad \forall u,c,n, \label{eq:sensing beamfoming}\\
		&\|\mathbf f_u[t_c^n]\|^2\le P_u^{\max},
		\quad \forall u,c,n,\label{eq:communication beamforming}\\
		&0\le \xi_u[t_c^n]\le 1,\quad \forall u,c,n,\label{eq:xi}\\
		&0\le \lambda_u[t_c^n]\le 1,\quad \forall u,c,n,\label{eq:lambda}\\
		&\alpha_{u,k}[t_c]\in\{0,1\},
		\quad \forall u,k,c,\label{eq:alpha}\\
		&\sum_{u=1}^{U}\alpha_{u,k}[t_c]\le 1, 
		\quad \forall k,c,\label{eq:total UAV alpha}\\
		&\sum_{k=1}^{K}\alpha_{u,k}[t_c]\le K_u^{\max},
		\quad \forall u,c,\label{eq:total target alpha}\\
		&(1-\lambda_u[t_c^n])C_u^{\rm sen}[t_c^n]
		\le
		C_{u,\rm loc}^{\rm comp}[t_c^n],
		\quad \forall u,c,n,\label{eq:local}\\
		&f_{u,0}[t_c^n]\ge 0,\quad
		\sum_{u=1}^{U}f_{u,0}[t_c^n]\le f_0^{\max},
		\quad \forall c,n,\label{eq:computation}\\
		&\mathbf q_u[t_1]=\mathbf q_{u,I},
		\quad \forall u,\label{eq:start}\\
		&\|\mathbf q_u[t_c^n]-\mathbf q_{u'}[t_c^n]\|\ge d_{\min},
		\quad u\ne u',\forall c,n,\label{eq:security}\\
		&
		\|\mathbf v_u[t_c^n]\|
		\le
		v_{\max},
		\quad \forall u,c,n, \label{eq:speed}\\
		&e_u[t_c^n]\le e_{\rm th},
		\quad \forall u,c,n, \label{eq:train}\\
		&\lim_{C\to+\infty}
		\frac{1}{CN}
		\sum_{c=1}^{C}
		\sum_{n=1}^{N}
		\mathbb E(Q_u^A[t_c^n])
		< \infty,\quad \forall u, \label{eq:UAV queue}\\
		&\lim_{C\to+\infty}
		\frac{1}{CN}
		\sum_{c=1}^{C}
		\sum_{n=1}^{N}
		\mathbb E(Q_u^V[t_c^n])
		< \infty,\quad \forall u. \label{eq:USV queue}
	\end{align}
\end{subequations}
where \eqref{eq:sensing} guarantees the sensing performance requirement, while \eqref{eq:sensing beamfoming} and \eqref{eq:communication beamforming} constrain the sensing and communication power, respectively. \eqref{eq:xi} and \eqref{eq:lambda} define the feasible ranges of time allocation and offloading ratio. \eqref{eq:alpha}–\eqref{eq:total target alpha} represent UAV-target association constraints. \eqref{eq:local} and \eqref{eq:computation} restrict the UAV and USV computation resources, respectively. \eqref{eq:start}–\eqref{eq:speed} ensure UAV mobility constraints, while \eqref{eq:train} and \eqref{eq:UAV queue}–\eqref{eq:USV queue} guarantee inference accuracy and queue stability.


\section{Problem Transformation}
\label{sec:problem}
Due to the time-averaged objective function and long-term constraints \eqref{eq:UAV queue} and \eqref{eq:USV queue}, it is difficult to solve the problem (P1) directly. Therefore, we introduce the Lyapunov optimization method to transform the problem. 
The Lyapunov optimization is employed to address the long-term stochastic problem. The collections of UAV data queues $Q^A_u[t_c^n]$ and USV data queues $Q^V_u[t_c^n]$ can be denoted as $Q^A[t_c^n] \triangleq \left[ Q_1^A[t_c^n], Q_2^A[t_c^n], \dots, Q_U^A[t_c^n] \right]$ and $Q^V[t_c^n] \triangleq \left[ Q_1^V[t_c^n], Q_2^V[t_c^n], \dots, Q_U^V[t_c^n] \right]$.
Let $\Theta(t_c^n) \triangleq \left\{ Q^{A}[t_c^n], Q^{V}[t_c^n] \right\}$, and then the Lyapunov function is 
\begin{equation}
	L ( \Theta( t_c^n ) )=\frac{1} {2} \sum_{u=1}^{U} [ Q^A_{u} [ t_c^n ] ]^{2}+\frac{1} {2} \sum_{u=1}^{U} [ Q^V_{u} [ t_c^n ] ]^{2}.
\end{equation}

It can be found that the smaller the value of $L ( \Theta( t_c^n ) )$ is, the less queue backlog in the queue is. Without loss of generality, we assume that $L ( \Theta( 0 ) )=0$. 
We define the Lyapunov drift as the expected change in the
Lyapunov function over one time slot \cite{Bi2021Lyapunov}, that is
\begin{equation}
	\Delta( \Theta( t_c^n ) )=E [L ( \Theta( \operatorname{next}(t_c^n) ) )-L ( \Theta( t_c^n ) ) | \Theta( t_c^n ) ].
\end{equation}

To ensure the stability of the system, we should minimize the Lyapunov drift at each time slot and maintain the queue backlog at a low level. Obviously, reducing the queue backlog generally requires higher transmission and computation rates, which inevitably increase energy consumption. Therefore, to minimize the total energy consumption while maintaining queue stability, we first define the drift-plus-penalty function to establish a connection between the Lyapunov function and energy consumption, which can be expressed as
\begin{equation}
	\Delta( \Theta( t_c^n ) ) +\! V \! \mathbb{E} [E_ {\rm sum}[t_c^n]| \Theta( t_c^n )],
\end{equation}
where $V$ denotes the trade-off between the total energy consumption and Lyapunov drift, reflecting whether the current system places more emphasis on energy minimization or system stability. Specifically, the larger the value of $V$, the closer the total energy consumption of the system is to the optimal value. Conversely, a lower $V$ value focuses on minimizing the Lyapunov drift.

Directly solving the optimization problem of the given expression incurs high computational complexity. Therefore, we further introduce an upper bound of the drift-plus-penalty, as shown in the following Theorem 1.
\begin{theorem}
	For $\forall \Theta(t_c^n), \forall t \in \mathcal{T}$, the drift-plus-penalty has the following upper bound
	\begin{equation}
		\begin{aligned} 
			{{}} & {{} {{} \Delta( \Theta( t_c^n ) )+V \mathbb{E} [ E_{\rm s u m} [ t_c^n ] | \Theta( t_c^n ) ] \leq R+V \mathbb{E} [ E_{\rm s u m} [ t_c^n ] | \Theta( t_c^n ) ]}} \\ {{}} & {{} {{}+\mathbb{E} \Big\{\sum_{u=1}^{U} Q^A_{u} [ t_c^n ] \Big( \lambda_u[t_c^n]C^{\rm sen}_u[t_c^n]-C_u^{\mathrm{com}}[t_c^n] \Big) | \Theta( t_c^n ) \Big\}}} \\ {{}} & {{} {{}+\mathbb{E} \Big\{\sum_{u=1}^{U} Q^V_{u} [ t_c^n ] \Big( C_u^{\mathrm{com}}[t_c^n]-C_{u,0}^{\rm comp}[t_c^n] \Big) | \Theta( t_c^n ) \Big\},}} 
		\end{aligned} 
	\end{equation}
\end{theorem}
where $R$ is a bounded constant. 

\textit{Proof:} A similar proof has been provided in our previous work \cite{YWang2026-1} and is omitted here for brevity.

Finally, according to the Lyapunov optimization method, the stochastic optimization P1 is transformed to focus on minimizing the Lyapunov drift-plus-penalty upper bound in each time slot, by disregarding the fixed term $R$, P1 can be formulated as follows
\begin{align}
	\text{(P2):}\quad & \min_{\Omega} 
	\mathcal{G}[t_c^n] \\
	\text{s.t.} \quad & \eqref{eq:sensing}-\eqref{eq:train}, \nonumber
\end{align}
where  the specific form of the objective function is
\begin{align}
		\mathcal{G}[t_c^n]=&V  E_{\rm s u m} [ t_c^n ]+  \sum_{u=1}^{U} Q^A_{u} [ t_c^n ] \Big( \lambda_u[t_c^n]C^{\rm sen}_u[t_c^n]-C_u^{\mathrm{com}}[t_c^n] \Big)\nonumber\\&+\sum_{u=1}^{U} Q^V_{u} [ t_c^n ] \Big( C_u^{\mathrm{com}}[t_c^n]-C_{u,0}^{\rm comp}[t_c^n] \Big).
\end{align}

\section{Lyapunov-Guided Dual-Timescale Hierarchical Soft Actor-Critic Algorithm}
\label{sec:algorithm}

For problem (P2), the resulting mixed-integer dual-timescale nonlinear optimization problem is highly coupled, where large-timescale decisions, including UAV association and trajectory planning, affect subsequent small-timescale resource allocation, while slot-level variables are jointly coupled through sensing, communication, computation, and queue evolution. 
To address this challenge, we propose an LTH-SAC framework, as illustrated in Fig.~\ref{fig4}. At the beginning of each frame \(t_c\), the high-level policy observes the frame-level state and outputs the frame-level action including UAV-target association, UAV position, and reference values of time allocation and offloading ratio. Since the high-level action contains discrete and continuous variables, SAC with Gumbel-Softmax is adopted to handle the hybrid action space. Then, given the high-level action within each frame, the low-level policy adjusts slot-level resource allocation, including time allocation, offloading ratio, beamforming, and USV computing resources. After the low-level policy executes the
slot-level actions for all \(N\) small slots in frame \(t_c\), the high-level policy receives an aggregated frame-level reward based on slot-level outcomes and updates the next frame state.

\begin{figure*}[t]
	\centering
	\includegraphics[width=\textwidth, trim=5pt 15pt 10pt 10pt, clip]{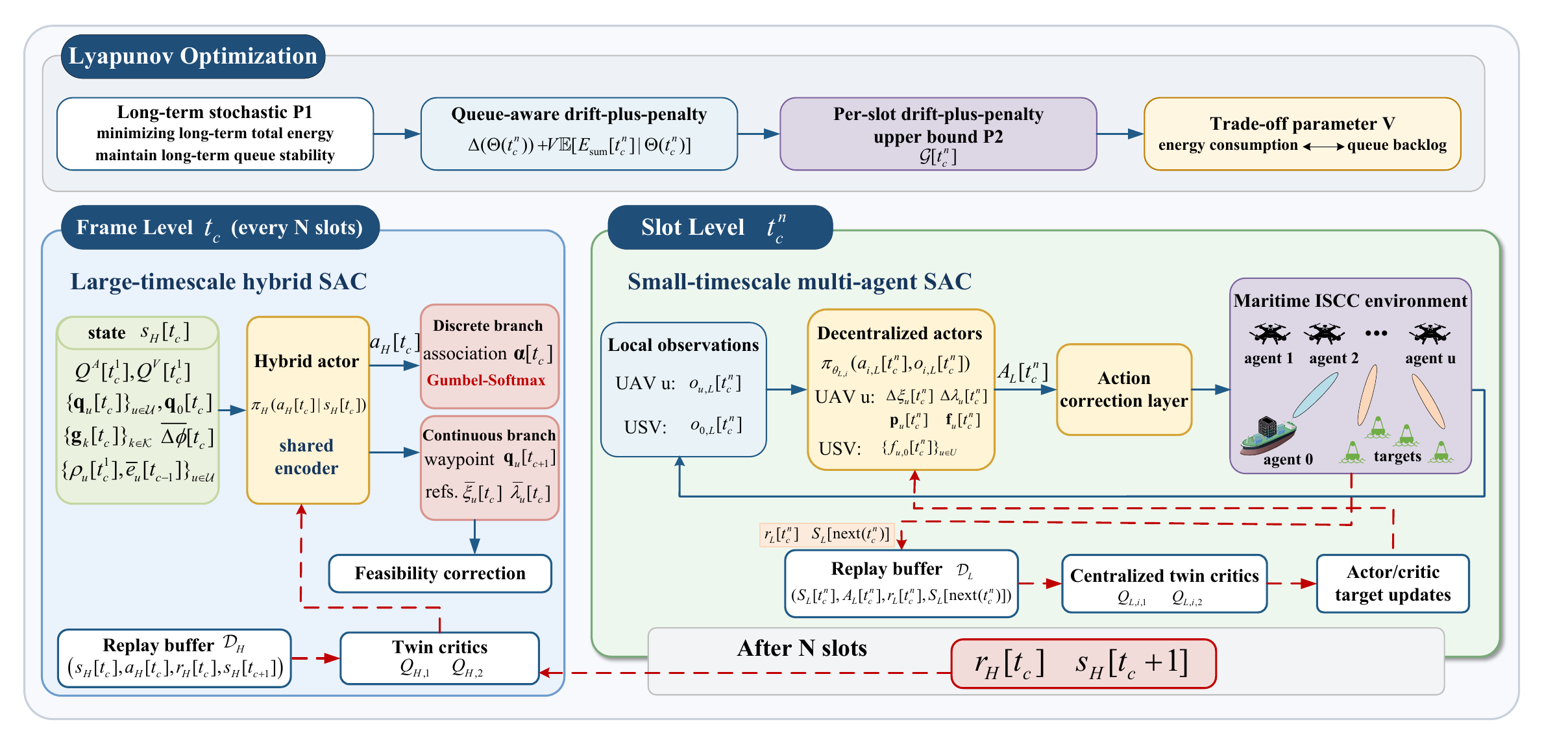}
	\vspace*{-0.05in}
	\caption{\small The proposed Lyapunov-guided dual-timescale hierarchical SAC framework for multi-UAV-assisted maritime ISCC systems.}
	\label{fig4}
\end{figure*}

\subsection{Large-Timescale Decision Based on Hybrid SAC}

To capture the long-term impact of UAV-target association, UAV trajectory,
and service-mode reference decisions under the dynamic maritime environment,
the high-level policy adopts a hybrid SAC framework with Gumbel-Softmax
reparameterization, whose state space, action space, and reward are defined as follows.

\subsubsection{State Space}For the high-level policy, 
the large-timescale state space in frame \(t_c\) can be defined as
\begin{equation}
	\begin{aligned}
		s_H[t_c]=
		\Big\{
		&Q^A[t_c^1],
		 Q^V[t_c^1], \mathbf q_0[t_c], \Delta\phi[t_c^1],\\
		&\{\mathbf q_u[t_c], \mathbf{g}_k[t_c],  \rho_u[t_c^1], \bar{e}_u[t_{c-1}] \}_{u\in\mathcal U, k\in\mathcal K} \Big\},
	\end{aligned}
\end{equation}
where $\bar{e}_u[t_{c-1}]$ is the statistical information classification-error state.

\subsubsection{Action Space}Corresponding to problem (P2), the high-level action space 
can be described by
\begin{equation}
	\begin{aligned}
		a_H[t_c]=
		\Big\{
		&\{\alpha_{u,k}[t_c],
		\mathbf q_u[t_{c+1}],
		\bar{\xi}_u[t_c],
	    \bar{\lambda}_u[t_c]\}_{u\in\mathcal U,k\in\mathcal K}
		\Big\},
	\end{aligned}
\end{equation}
where $\bar{\xi}_u[t_c]$ and $\bar{\lambda}_u[t_c]$ are reference values of the time allocation factor and offloading ratio, respectively.

Since the high-level action contains both discrete and continuous variables,
the high-level policy is factorized as
\begin{equation}
	\pi_H(a_H[t_c]|s_H[t_c])
	=
	\pi_H^{d}(\boldsymbol{\alpha}[t_c]|s_H[t_c])
	\pi_H^{c}(\mathbf a_H^{c}[t_c]|s_H[t_c]),
\end{equation}
where $\boldsymbol{\alpha}[t_c]=\{\alpha_{u,k}[t_c]\}_{u\in\mathcal U,k\in\mathcal K}$ denotes the discrete
association action, and
	$\mathbf a_H^{c}[t_c]
	=
	\left\{
	\{\mathbf q_u[t_{c+1}]\},
	\{\bar{\xi}_u[t_c]\},
	\{\bar{\lambda}_u[t_c]\}
	\right\}$
denotes the continuous high-level action.

For the discrete UAV-target association, the high-level actor outputs logits
\(\ell_{k,j}[t_c]\) for assigning target \(k\) to option
\(j\in\{0,\cdots,U\}\), where \(j=0\) means that target \(k\) is not assigned
to any UAV. The Gumbel-Softmax reparameterization is used to generate a
differentiable approximation of the discrete association decision:
\begin{equation}
	\tilde y_{k,j}[t_c]
	=
	\frac{
		\exp\left((\ell_{k,j}[t_c]+g_{k,j})/\tau_g\right)
	}{
		\sum_{j'=0}^{U}
		\exp\left((\ell_{k,j'}[t_c]+g_{k,j'})/\tau_g\right)
	},
\end{equation}
where \(g_{k,j}=-\log(-\log \epsilon_{k,j})\),
\(\epsilon_{k,j}\sim \mathrm{Uniform}(0,1)\), and \(\tau_g\) is the Gumbel
temperature. During training, \(\tilde y_{k,j}[t_c]\) is used for
differentiable backpropagation.
During execution, the discrete association is obtained by selecting the maximum probability output and then corrected to satisfy association constraints.
The continuous high-level action is generated by the reparameterization trick
in SAC, and the raw continuous output is projected into the feasible region to satisfy constraints.

\subsubsection{Reward}We define the high-level reward for executing the action $a_H[t_c]$ based
on the state $s_H[t_c]$ during frame \(t_c\) as
\begin{equation}
	r_H[t_c]
	=
	\sum_{n=1}^{N} r_L[t_c^n]
	-
	\mathcal P_H[t_c],
\end{equation}
where \(\mathcal P_H[t_c]\) denotes the penalty for violating large-timescale
constraints, including the UAV-target association constraints, UAV speed
constraints, and collision avoidance constraints.

\subsubsection{Neural Networks}In the high-level SAC, two soft Q-functions \(Q_{H,1}\) and \(Q_{H,2}\) are
used to reduce the overestimation bias. The high-level transition is stored as $\left(s_H[t_c],a_H[t_c],r_H[t_c],s_H[t_{c+1}]\right)$ in the high-level replay buffer \(\mathcal D_H\). For a mini-batch sampled
from \(\mathcal D_H\), the high-level soft target value is given by
\begin{equation}
	\begin{aligned}
		y_H
		=
		r_H[t_c]
		+
		\gamma_H
		\Big[
		&\min_{i=1,2}
		Q_{H,i}^{-}(s_H[t_{c+1}],a_H'[t_{c+1}])\\
		&-
		\eta_H
		\log\pi_{H,\theta}(a_H'[t_{c+1}]|s_H[t_{c+1}])
		\Big],
	\end{aligned}
\end{equation}
where \(a_H'[t_{c+1}]\sim\pi_H(\cdot|s_H[t_{c+1}])\),
\(Q_{H,i}^{-}\) is the target Q-network with parameter $\psi_{H,i}^{-}$, \(\gamma_H\) is the high-level
discount factor, and \(\eta_H\) is the entropy temperature. The high-level Q-networks are updated by minimizing the temporal-difference error between the estimated Q-values and the soft target values, which is given by
\begin{equation}
	J_Q^H(\psi_{H,i})
	=
	\mathbb E_{\mathcal D_H}
	\left[
	\left(
	Q_{H,i}(s_H[t_c],a_H[t_c];\psi_{H,i})-y_H
	\right)^2
	\right],
\end{equation}
where \(i\in\{1,2\}\), and \(\psi_{H,i}\) denotes the parameter of the
\(i\)-th high-level Q-network.\\
The high-level policy network is updated by minimizing the entropy-regularized
policy objective, which is given by
\begin{equation}
	\begin{aligned}
		J_{\pi}^H(\theta_H)
		=
		\mathbb E_{s_H[t_c]\sim\mathcal D_H}
		\Big[
		&\eta_H
		\log\pi_H(a_H[t_c]|s_H[t_c])\\
		&-
		\min_{i=1,2}Q_{H,i}(s_H[t_c],a_H[t_c])
		\Big],
	\end{aligned}
\end{equation}
where \(\theta_H\) denotes the parameter of the high-level policy network. The target Q-networks are then updated by the following rule 
\begin{equation}
	\psi_{H,i}^- \leftarrow
	\tau_H\psi_{H,i}+(1-\tau_H)\psi_{H,i}^-,
	\quad i\in\{1,2\},
\end{equation}
where \(\tau_H\in[0,1]\) is the soft update coefficient.
\subsection{Small-Timescale Resource Allocation Based on MASAC}

Due to the continuous nature and high-dimensional coupling characteristics of slot-level decision variables, we adopt a centralized training with decentralized execution framework and design a cooperative scheme based on MASAC. During the training phase, all UAV agents are jointly optimized by leveraging global state information. In the execution phase, each agent independently outputs its actions only according to its local observation.
Specifically, the set of low-level agents is defined as
$\mathcal I_L=\{0,1,\cdots,U\}$,
where \(u\in\{1,\cdots,U\}\) denotes UAV agent $u$, and \(0\)
denotes the USV agent. 

\subsubsection{Environment State Space}The global state space of the low-level multi-agent SAC in
slot \(t_c^n\) includes queues, UAV-USV channel and the USV roll angle information, which is given by
\begin{equation}
	\begin{aligned}
		S_L[t_c^n]=
		\Big\{
		&a_H[t_c],
		Q^A[t_c^n],
	    Q^V[t_c^n], \{\mathbf h_u[t_c^n]\}_{u\in\mathcal U},
		\Delta\phi[t_c^n],\\
		&\{\mathbf q_u[t_c^n], \rho_u[t_c^n],\theta_{u,k}[t_c^n]\}_{u\in\mathcal U,k\in\mathcal K}
		\Big\},
	\end{aligned}
\end{equation}
where \(a_H[t_c]\) is the high-level action.

\subsubsection{Observation Space}In the partially observable ISCC
environment, local observations for UAV agent $u$ and the USV
agent in the slot $t^n_c$ are detailed below.

For the UAV agent \(u\), it observes the states of itself and targets, as well as the roll angle state of the USV, i.e.,
\begin{equation}
	\begin{aligned}
		o_{u,L}[t_c^n]=
		\Big\{
		&\{\alpha_{u,k}[t_c]\}_{k\in\mathcal K},
		\mathbf q_u[t_c^n], 
		Q_u^A[t_c^n],
		Q_u^V[t_c^n],\mathbf h_u[t_c^n],\\
		&
		\Delta\phi[t_c^n],
		\{\theta_{u,k}[t_c^n], d_{u,k}[t_c^n]]\}_{k\in\mathcal K},
		\bar{\xi}_u[t_c],
		\bar{\lambda}_u[t_c]
		\Big\}.
	\end{aligned}
\end{equation}

For the USV agent, it observes the states of itself and UAVs, i.e.,
\begin{equation}
	\begin{aligned}
		o_{0,L}[t_c^n]
		=
		\Big\{
		&a_H[t_c],
        Q^A[t_c^n],
	    Q^V[t_c^n],
		\{\mathbf h_u[t_c^n]\}_{u\in\mathcal U},
		\Delta\phi[t_c^n]
		\Big\}.
	\end{aligned}
\end{equation}
\subsubsection{Action Space}
Based on the associated observational data $\{o_{u,L}[t_c^n], o_{0,L}[t_c^n]\}$, each UAV and the USV make decisions of action. In the time slot $t_c^n$, the UAV agent \(u\) outputs the action
\begin{equation}
	a_{u,L}[t_c^n]
	=
	\left\{
	\Delta\xi_u[t_c^n],
	\Delta\lambda_u[t_c^n],
	\mathbf p_u[t_c^n],
	\mathbf f_u[t_c^n]
	\right\},
\end{equation}
and the USV agent outputs the computing resource allocation action
\begin{equation}
	a_{0,L}[t_c^n]=
	\{f_{u,0}[t_c^n]\}_{u\in\mathcal U}.
\end{equation}
Accordingly, the joint low-level action is denoted by
\begin{equation}
	A_L[t_c^n]
	=
	\left\{
	a_{0,L}[t_c^n],
	a_{1,L}[t_c^n],
	\cdots,
	a_{U,L}[t_c^n]
	\right\}.
\end{equation}

Given the high-level reference values \(\bar{\xi}_u[t_c]\) and
\(\bar{\lambda}_u[t_c]\), the actual slot-level time allocation factor and
offloading ratio are reconstructed as $\xi_u[t_c^n]
	=
	\Pi_{[0,1]}
	\left(
	\bar{\xi}_u[t_c]+\Delta\xi_u[t_c^n]
	\right),$
$
	\lambda_u[t_c^n]
	=
	\Pi_{[\lambda_u^{\min}[t_c^n],1]}
	\left(
	\bar{\lambda}_u[t_c]+\Delta\lambda_u[t_c^n]
	\right).
$


\subsubsection{Reward}According to the Lyapunov-guided objective in (P2), we define a
shared cooperative reward for all low-level agents as
\begin{equation}
	r_L[t_c^n]
	=
	-\mathcal{G}[t_c^n]
	-
	\mathcal P_L[t_c^n],
\end{equation}
where \(\mathcal P_L[t_c^n]\) denotes the penalty for violating
small-timescale constraints, including the transmit power constraints, local
computing constraints, sensing-rate constraints, classification-error
constraints, and USV computing resource constraints. By using a shared reward,
all low-level agents are encouraged to cooperatively reduce the system energy
consumption while maintaining queue stability.

\subsubsection{Neural Networks}Each low-level agent \(i\in\mathcal I_L\) maintains a stochastic
actor network
\(\pi_{\theta_{L,i}}(a_{i,L}[t_c^n]|o_{i,L}[t_c^n])\), parameterized by
\(\theta_{L,i}\), for generating slot-level actions from local observations.
Meanwhile, two centralized soft Q-networks
\(Q_{L,i,1}(S_L[t_c^n],A_L[t_c^n];\psi_{L,i,1})\) and
\(Q_{L,i,2}(S_L[t_c^n],A_L[t_c^n];\psi_{L,i,2})\) are employed to evaluate the
global state-action value and mitigate overestimation bias. During training,
all agents collect their local experience tuple
to form the global transition
$\left\langle
	S_L[t_c^n],
	A_L[t_c^n],
	r_L[t_c^n],
	S_L[\mathrm{next}(t_c^n)]
	\right\rangle,$
which is stored in the low-level replay buffer \(\mathcal D_L\).

For a mini-batch sampled from \(\mathcal D_L\), the target value of the
\(i\)-th low-level agent is given by
\begin{equation}
	\begin{aligned}
		y_{L,i}
		=
		r_L[t_c^n]
		+
		\gamma_L
		\Big[
		&\min_{\ell=1,2}
		Q_{L,i,\ell}^{-}
		\left(S_L',A_L';\psi_{L,i,\ell}^{-}\right)\\
		&-
		\eta_L
		\sum_{j\in\mathcal I_L}
		\log
		\pi_{\theta_{L,j}}
		\left(a_{j,L}'|o_{j,L}'\right)
		\Big],
	\end{aligned}
\end{equation}
where \(S_L'=S_L[\mathrm{next}(t_c^n)]\),
\(A_L'=\{a_{j,L}'\}_{j\in\mathcal I_L}\),
\(a_{j,L}'\sim\pi_{\theta_{L,j}}(\cdot|o_{j,L}')\),
\(Q_{L,i,\ell}^{-}\) is the target Q-network, \(\gamma_L\) is the low-level
discount factor, and \(\eta_L\) is the entropy temperature.

Then, the low-level Q-networks are updated by minimizing the
temporal-difference error between the estimated Q-value and the target value,
which is given by
\begin{equation}
	\resizebox{1\hsize}{!}{$\begin{aligned}
		J_Q^{L,i}(\psi_{L,i,\ell})
		=
		\mathbb E_{\mathcal D_L}
		\left[
		\left(
		Q_{L,i,\ell}(S_L[t_c^n],A_L[t_c^n];\psi_{L,i,\ell})
		-
		y_{L,i}
		\right)^2
		\right],
		\end{aligned}$}
\end{equation}
where $\ell\in\{1,2\}$.

The low-level policy network of agent \(i\) is updated by minimizing the
entropy-regularized policy objective, which is given by
\begin{equation}
	\begin{aligned}
		J_{\pi}^{L,i}(\theta_{L,i})
		=
		\mathbb E_{\mathcal D_L}
		\Big[
		&\eta_L
		\log
		\pi_{\theta_{L,i}}
		\left(a_{i,L}[t_c^n]|o_{i,L}[t_c^n]\right)\\
		&-
		\min_{\ell=1,2}
		Q_{L,i,\ell}
		\left(S_L[t_c^n],A_L[t_c^n]\right)
		\Big].
	\end{aligned}
\end{equation}\\

Finally, the target Q-networks are updated by the soft update rule as follows
\begin{equation}
	\psi_{L,i,\ell}^{-}
	\leftarrow
	\tau_L\psi_{L,i,\ell}
	+
	(1-\tau_L)\psi_{L,i,\ell}^{-},
	\quad i\in\mathcal I_L,\ \ell\in\{1,2\},
\end{equation}
where \(\tau_L\in[0,1]\) is the soft update coefficient.

Based on the proposed LTH-SAC framework, the interaction between
the large-timescale and small-timescale policies is established through
hierarchical decision-making and information exchange, as shown in
Algorithm~\ref{alg:LTHSAC}. At each frame, the high-level hybrid SAC observes the
frame-level state and generates feasible association, waypoint, and
reference decisions. Guided by these decisions, the low-level MASAC
agents perform slot-level resource allocation based on local observations.
The environment then returns rewards and updates system states.
Slot-level transitions are stored in the low-level replay buffer for
MASAC training, while the accumulated frame-level reward is used to
update the high-level SAC after completing all slots within a frame.

\begin{algorithm}[t]
	\caption{Proposed LTH-SAC training framework}
	\label{alg:LTHSAC}
	\small
	\begin{algorithmic}[1]
		\State Initialize high-level SAC networks and low-level SAC networks;
		\State Initialize target networks and replay buffers $\mathcal{D}_H$ and $\mathcal{D}_L$;
		
		\For{episode $=1,2,\cdots$}
		\State Initialize UAV positions, queue states, channel states, target \Statex \hspace{\algorithmicindent}states, and USV roll states;
		
		\For{frame $c=1,2,\cdots,C$}
		\State Observe high-level state $s_H[t_c]$;
		\State Select high-level action $a_H[t_c]$ according to
		\Statex \hspace{3em}$\pi_H(\cdot|s_H[t_c])$;
		\State Evaluate high-level penalty $P_H[t_c]$;
		\State Correct $a_H[t_c]$ to obtain feasible UAV-target association, \Statex \hspace{3em}next UAV waypoint, and reference values $\bar{\xi}_u[t_c]$, $\bar{\lambda}_u[t_c]$;
		\State Set accumulated high-level reward $r_H[t_c]\leftarrow 0$;
		
		\For{slot $n=1,2,\cdots,N$}
		\State Observe low-level state $S_L[t_c^n]$ under the guidance of \Statex \hspace{4.5em}$a_H[t_c]$;
		\State For each low-level agent $i\in\mathcal I_L$, select 
		$a_{i,L}[t_c^n]$ \Statex \hspace{4.5em}according to 
		$\pi_{\theta_{L,i}}(\cdot|o_{i,L}[t_c^n])$;
		\State Reconstruct feasible executable variables $\xi_u[t_c^n]$, \Statex \hspace{4.5em}$\lambda_u[t_c^n]$, $\mathbf{p}_u[t_c^n]$, $\mathbf{f}_u[t_c^n]$, and $f_{u,0}[t_c^n]$;
		\State Execute the action and evaluate system performance;
		\State Calculate low-level reward $r_L[t_c^n]$ and update queue \Statex \hspace{4.5em}states;
		\State Store low-level transition in $\mathcal{D}_L$;
		\State Update low-level SAC networks;
		\State Accumulate reward $r_H[t_c]\leftarrow r_H[t_c]+r_L[t_c^n]$;
		\EndFor
		
		\State $r_H[t_c]\leftarrow r_H[t_c]-P_H[t_c]$;
		\State Obtain next high-level state $s_H[t_{c+1}]$;
		\State Store high-level transition in $\mathcal{D}_H$;
		\State Update high-level SAC networks;
		\EndFor
		\EndFor
	\end{algorithmic}
\end{algorithm}

\section{Simulation Results}
\label{sec:results}
We consider a maritime ISCC system consisting of one USV,
$U=6$ UAVs, and $K=10$ stationary maritime targets within a
$1,000~\mathrm{m}\times1,000~\mathrm{m}$ area. UAVs fly at a fixed altitude with a maximum speed of
$v_{\rm max}=20~\mathrm{m/s}$ and a minimum safety distance of
$d_{\rm min}=25~\mathrm{m}$. 
Each UAV is equipped with four transmit antennas and four receive antennas.
Each episode contains 40 large-timescale frames, each consisting of five
time slots with a duration of $1~\mathrm{s}$.
The carrier frequency $f_c$ and communication bandwidth $B_u$ are set to
$5.8~\mathrm{GHz}$ and $1~\mathrm{MHz}$, respectively.
Both the maximum sensing power and communication power of each UAV are
$2~\mathrm{W}$.
The UAV and USV computing capabilities are set to $1~\mathrm{GHz}$ and $5~\mathrm{GHz}$, respectively.
To characterize the maritime environment, the USV roll motion is
generated using 32 irregular-wave components with a significant wave
height of $0.7~\mathrm{m}$.
The remaining key simulation parameters are summarized in
Table~\ref{tab:simulation_parameters} \cite{Zhang2019UAVMEC}.

For the proposed LTH-SAC algorithm, both the high-level and low-level
Actor and Critic adopt two-layer fully connected networks, with each
hidden layer containing 512 neurons.
The high-level and low-level discount factors are set to 0.98 and 0.95,
respectively.
The replay buffer capacity and batch size are
$3\times10^{5}$ and 256, respectively.
The proposed LTH-SAC algorithm is compared with SAC, deep deterministic policy gradient (DDPG), and the greedy
algorithm \cite{Xu2022MultiUAV}.
\begin{table}[t]
	\caption{Simulation Parameters}
	\label{tab:simulation_parameters}
	\centering
	\renewcommand{\arraystretch}{1.15}
	\small
	\begin{tabular}{p{0.68\columnwidth}p{0.22\columnwidth}}
		\hline
		\textbf{Parameter} & \textbf{Value} \\
		\hline
		USV antenna height $h_r$
		& $5\,\mathrm{m}$ \\
		
		UAV height $H_u$
		& $100\,\mathrm{m}$\\
		
		Normalized sea-clutter variance $\sigma_{u,z}^2$
		& 0.02 \\
		
		Normalized radar-noise variance $\sigma_{u,r}^2$
		& $10^{-3}$ \\
		
		Communication noise power $\sigma_{u,c}^2$
		& $2\times10^{-14}$\\
		
		Computation energy coefficient $\kappa_u,\kappa_0$
		& $10^{-27}$ \\
		
		Data size per sensing sample $D_s$
		& 10 bits/sample \\
		
		CPU cycles required per bit $(l_u,l_0)$
		& (500, 500)\\
		
		
		Natural angular frequency $\omega_{\Delta\phi}$
		& 1 \rm{rad/s}\\
		
		Damping ratio $\zeta_{\Delta\phi}$
		& 0.12 \\
		
		
		
		
		Mean induced velocity $v_{0}$
		& 4.03 \rm{m/s} \\
		
		Fuselage drag ratio $d_{0}$
		& 0.6\\
		
		Rotor solidity $s$
		& 0.05 \\
		
		Air density $\rho_0$
		& 1.225 $\rm {kg/m^{3}}$\\
		
		Rotor-disc area $A$
		& 0.503 $\mathrm{m}^{2}$ \\
		
		Target-network update factors $(\tau_H,\tau_L)$
		& (0.002, 0.005) \\
		\hline
	\end{tabular}
\end{table}

Fig.~\ref{fig:reward_convergence} evaluates the reward convergence of the algorithms. It can be observed that LTH-SAC gradually improves average reward during early exploration and converges after about 100 episodes to a higher average reward than the benchmarks, with a relatively narrow variance band. This indicates stable training and can be attributed to the hierarchical two-timescale architecture, which effectively coordinates long-timescale decisions with short-timescale resource allocation. Overall, LTH-SAC effectively handles the coupled dual-timescale decisions in the considered multi-UAV maritime ISCC system.

Fig.~\ref{fig:queue_stability} illustrates the queue stability index during training, where a smaller value indicates better queue stability. Although all learning-based algorithms fluctuate during early exploration due to policy exploration and dynamic resource adjustment, the queue stability index of LTH-SAC decreases rapidly and remains at a relatively lower level after convergence than the benchmarks. Fig.~\ref{fig:queue_backlog} presents the queue backlog evolution of UAVs and the USV under the proposed LTH-SAC algorithm. We can observe that all the queue backlogs gradually converge to bounded values after approximately 100 episodes without persistent accumulation. These results demonstrate that LTH-SAC effectively coordinates sensing, communication offloading, and computation processing to maintain favorable system queue stability.

\begin{figure*}[t]
	\centering
	
	\subfloat[ Convergence performance of all algorithms.]
	{
		\includegraphics[width=0.3\textwidth]{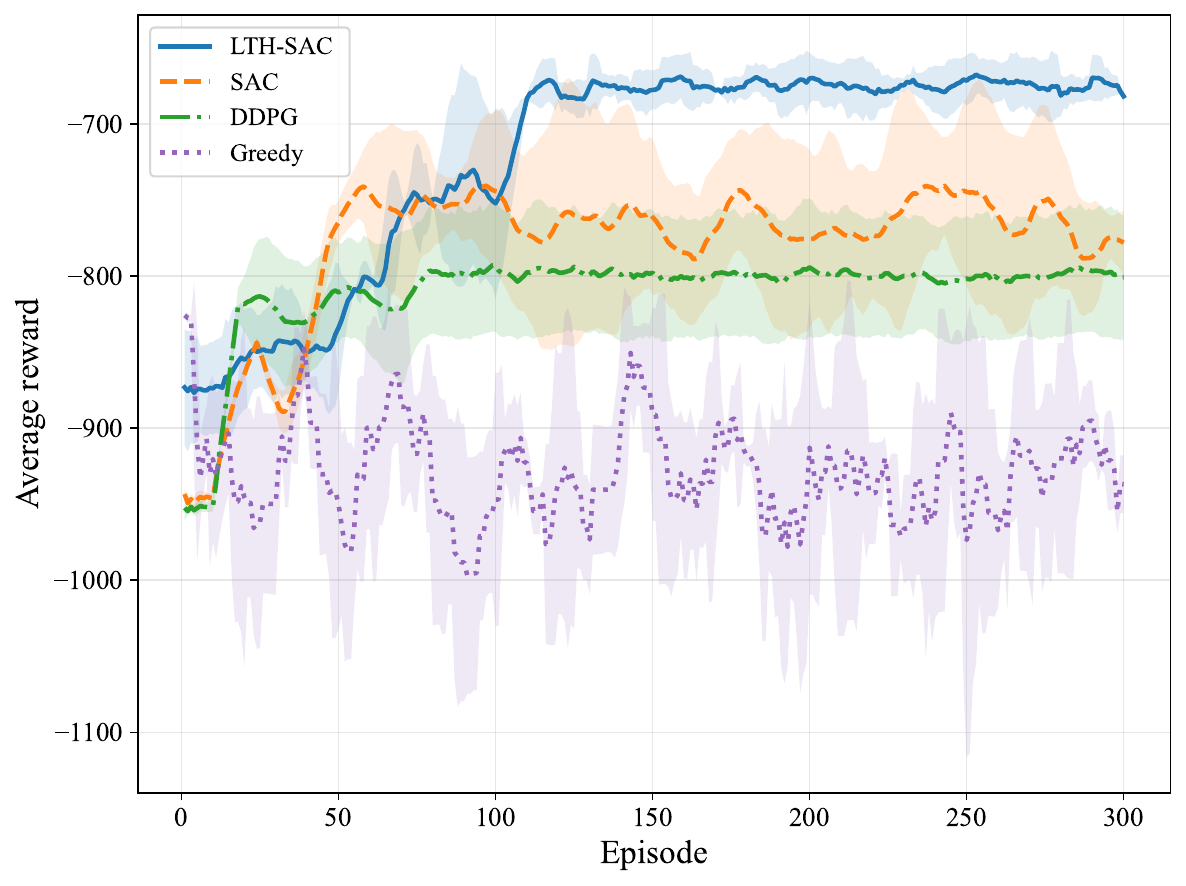}
		\label{fig:reward_convergence}
	}
	\hfill
	\subfloat[Queue stability index of all algorithms.]
	{
		\includegraphics[width=0.3\textwidth]{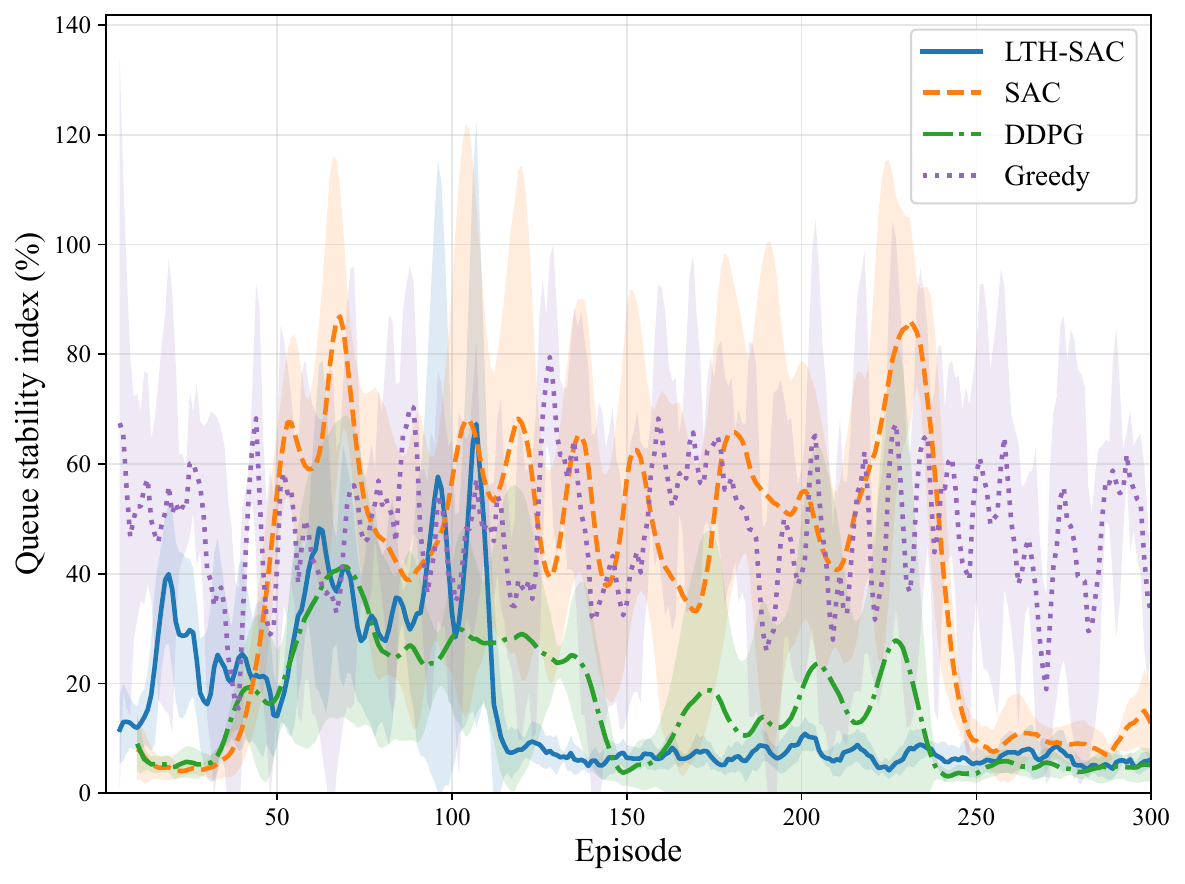}
		\label{fig:queue_stability}
	}
	\hfill
	\subfloat[Queue stability of UAVs and USV.]
	{
		\includegraphics[width=0.3\textwidth]{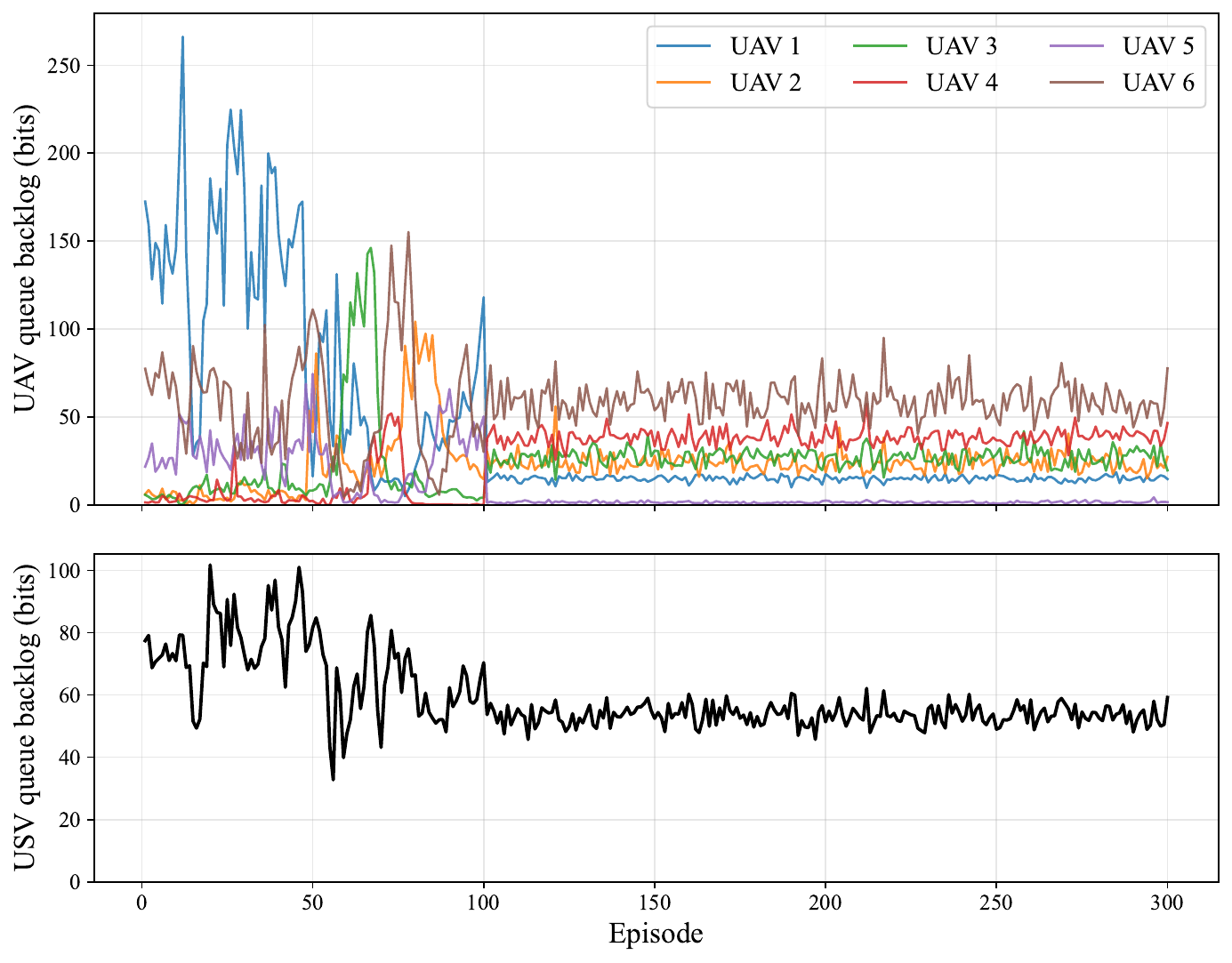}
		\label{fig:queue_backlog}
	}
	
	\caption{Convergence and queue stability performance of different algorithms.}
	\label{fig:convergence_queue}
\end{figure*}


\begin{figure*}[t]
	\centering
	
	\subfloat[Impact of $\xi_u$ on performance trade-off.]
	{
		\includegraphics[width=0.3\textwidth, trim=0pt 0pt 0pt 0pt, clip]{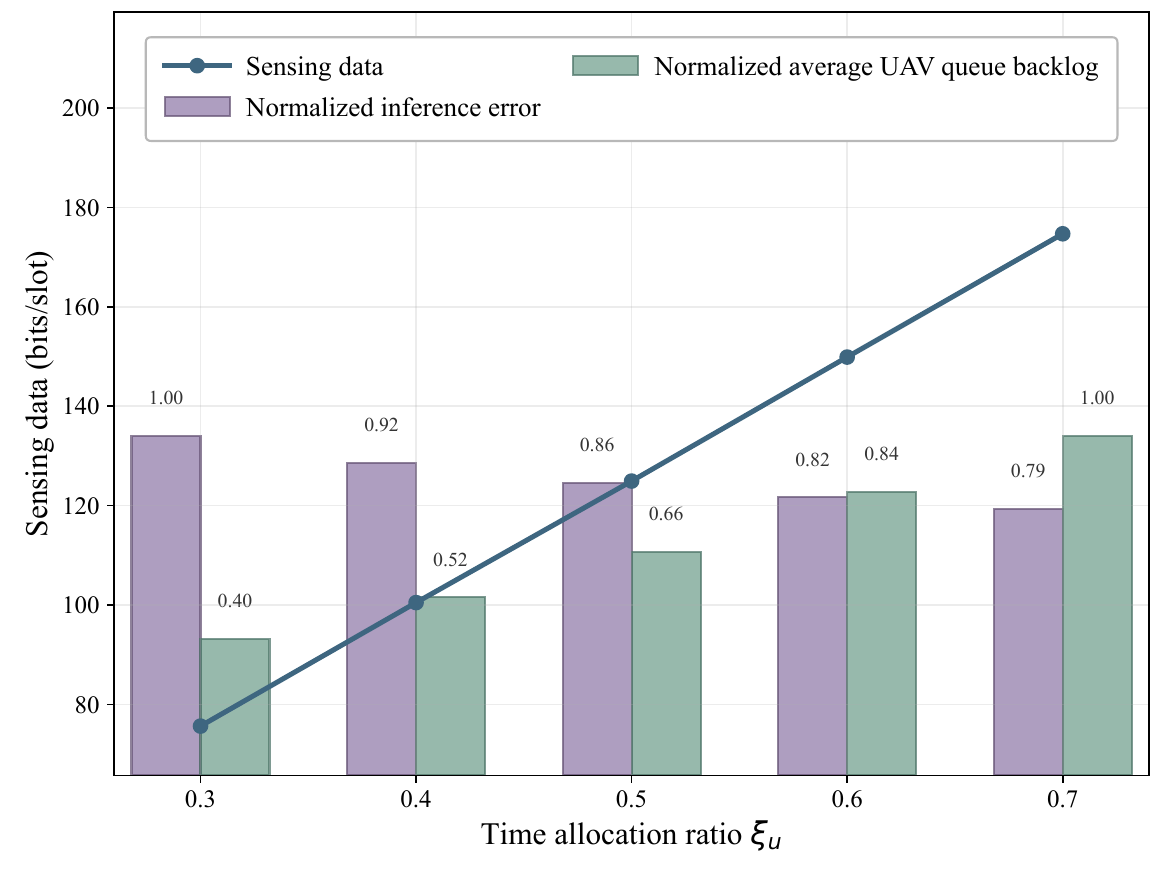}
		\label{fig:xi}
	}
	\hfill
	\subfloat[Impact of $\lambda_u$ on computing energy.]
	{
		\includegraphics[width=0.3\textwidth, trim=0pt 0pt 0pt 0pt, clip]{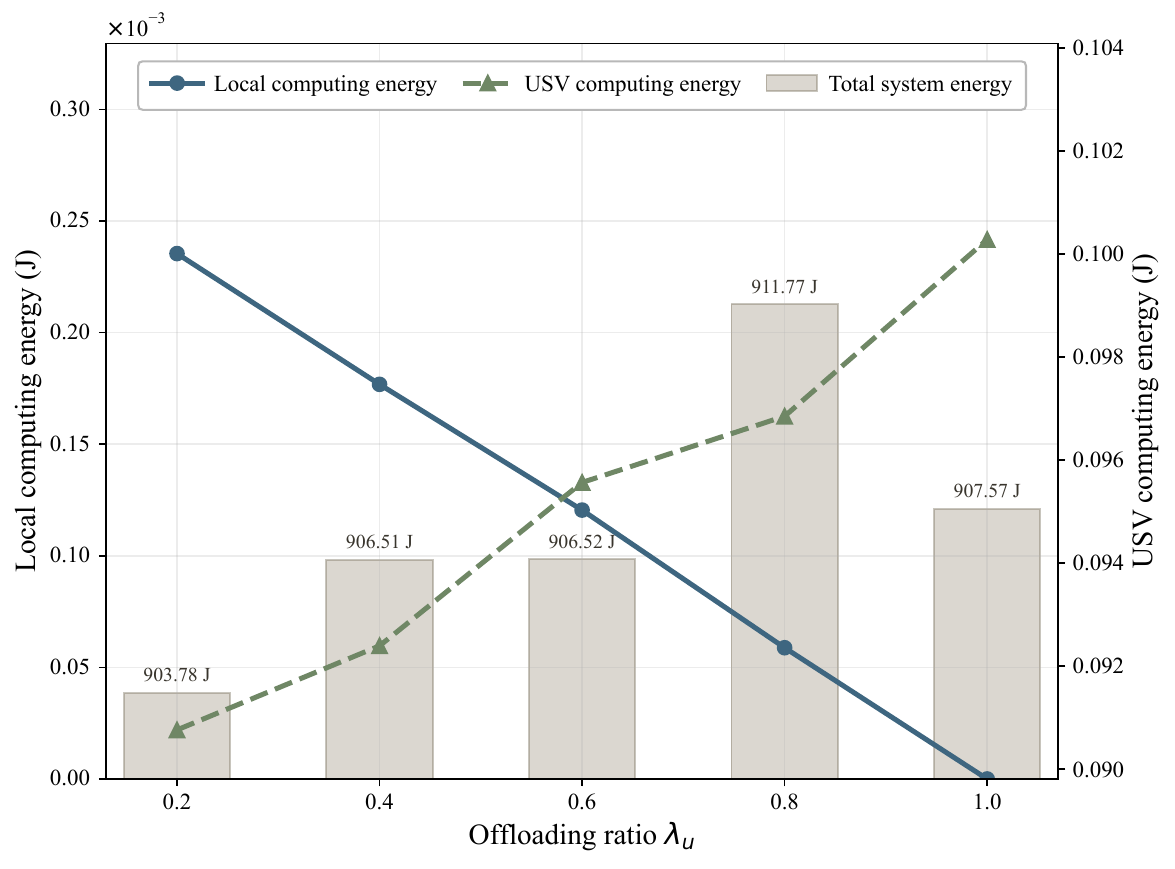}
		\label{fig:lambda_energy_tradeoff}
	}
	\hfill
	\subfloat[Impact of $\lambda_u$ on queue backlog.]
	{
		\includegraphics[width=0.3\textwidth, trim=0pt 0pt 0pt 0pt, clip]{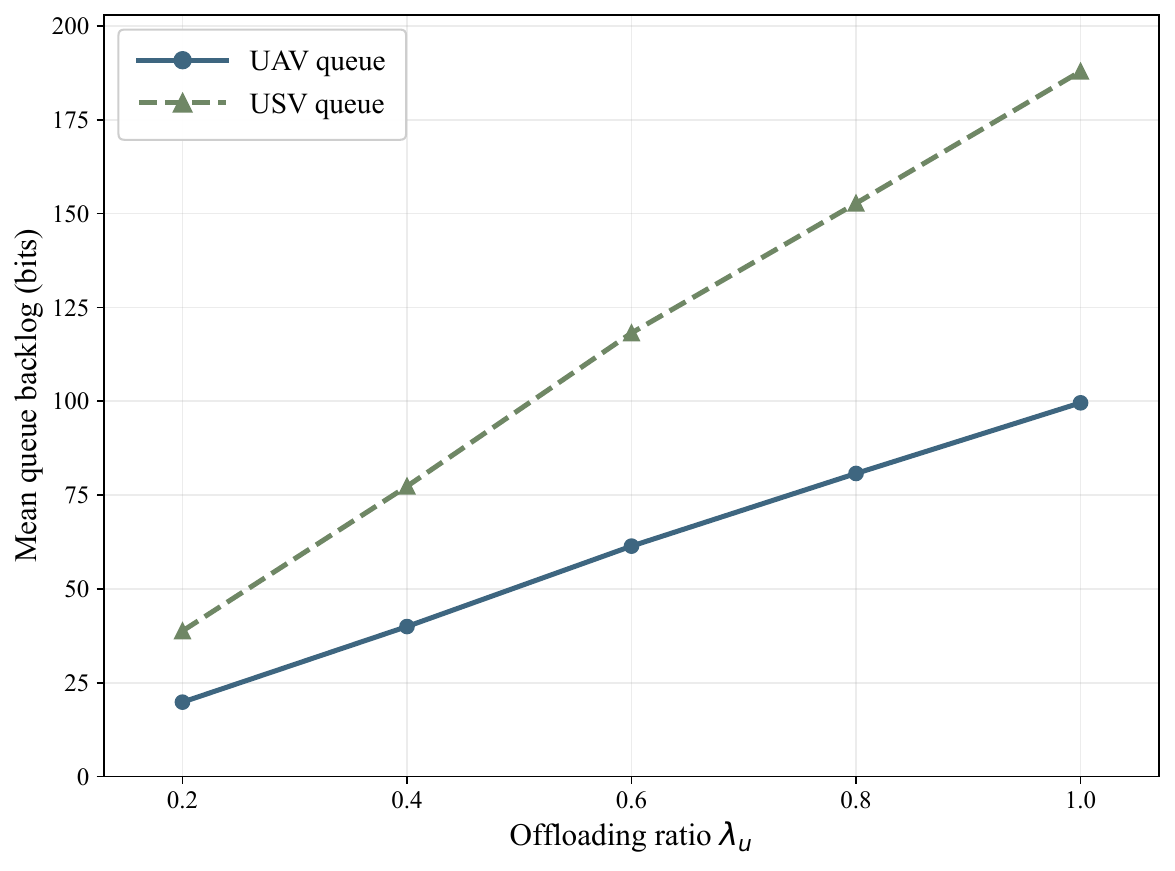}
		\label{fig:lambda_queue_tradeoff}
	}
	
	\caption{Performance trade-off analysis under different system parameters.}
	\label{fig:tradeoff}
	
	\vspace{-2mm}
\end{figure*}
\begin{figure}[t]
	\centering
	\includegraphics[width=0.35\textwidth, trim=0pt 0pt 0pt 0pt, clip]{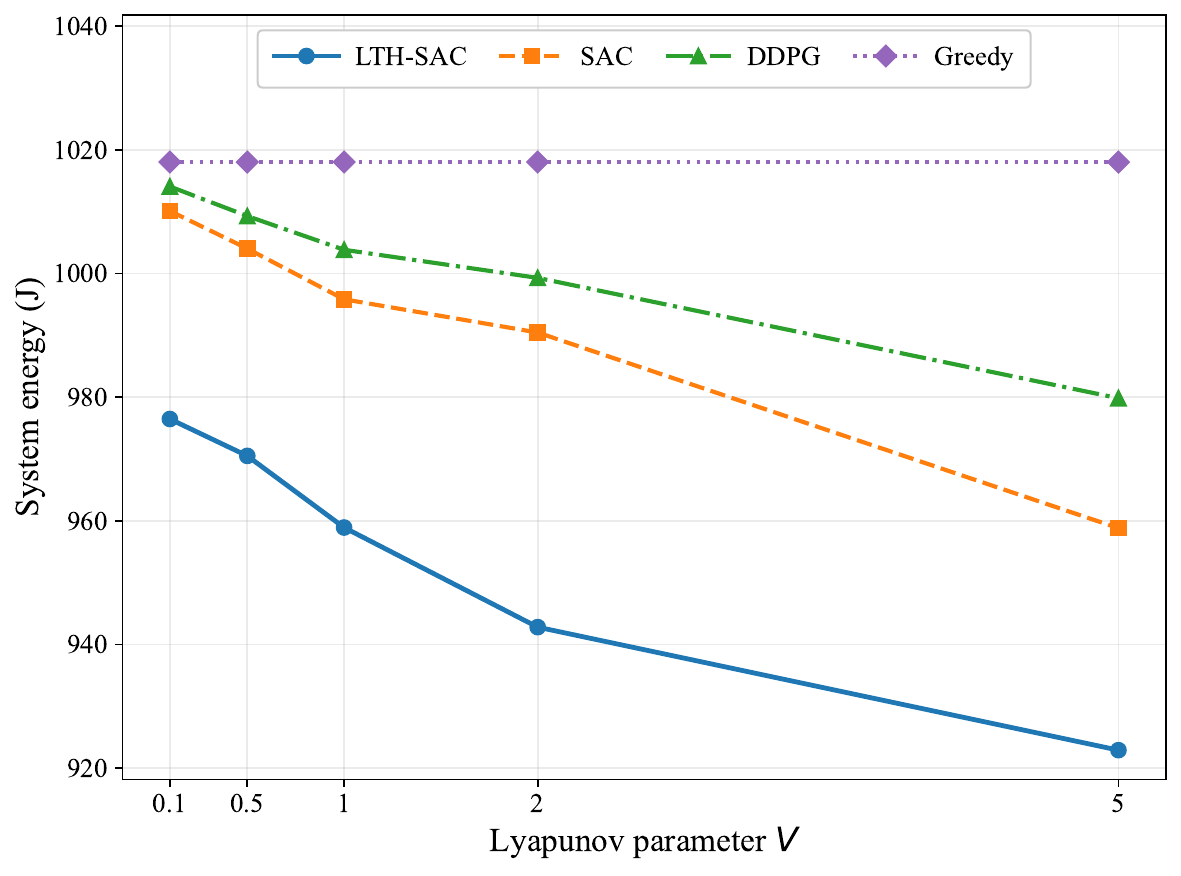}
	\caption{\small System energy under different Lyapunov control parameters.}\label{fig:V}
	\vspace{-1mm}
\end{figure}
\begin{figure}[t]
	\centering
	\subfloat[Different numbers of sensing targets.]
	{
		\includegraphics[width=0.8\linewidth]{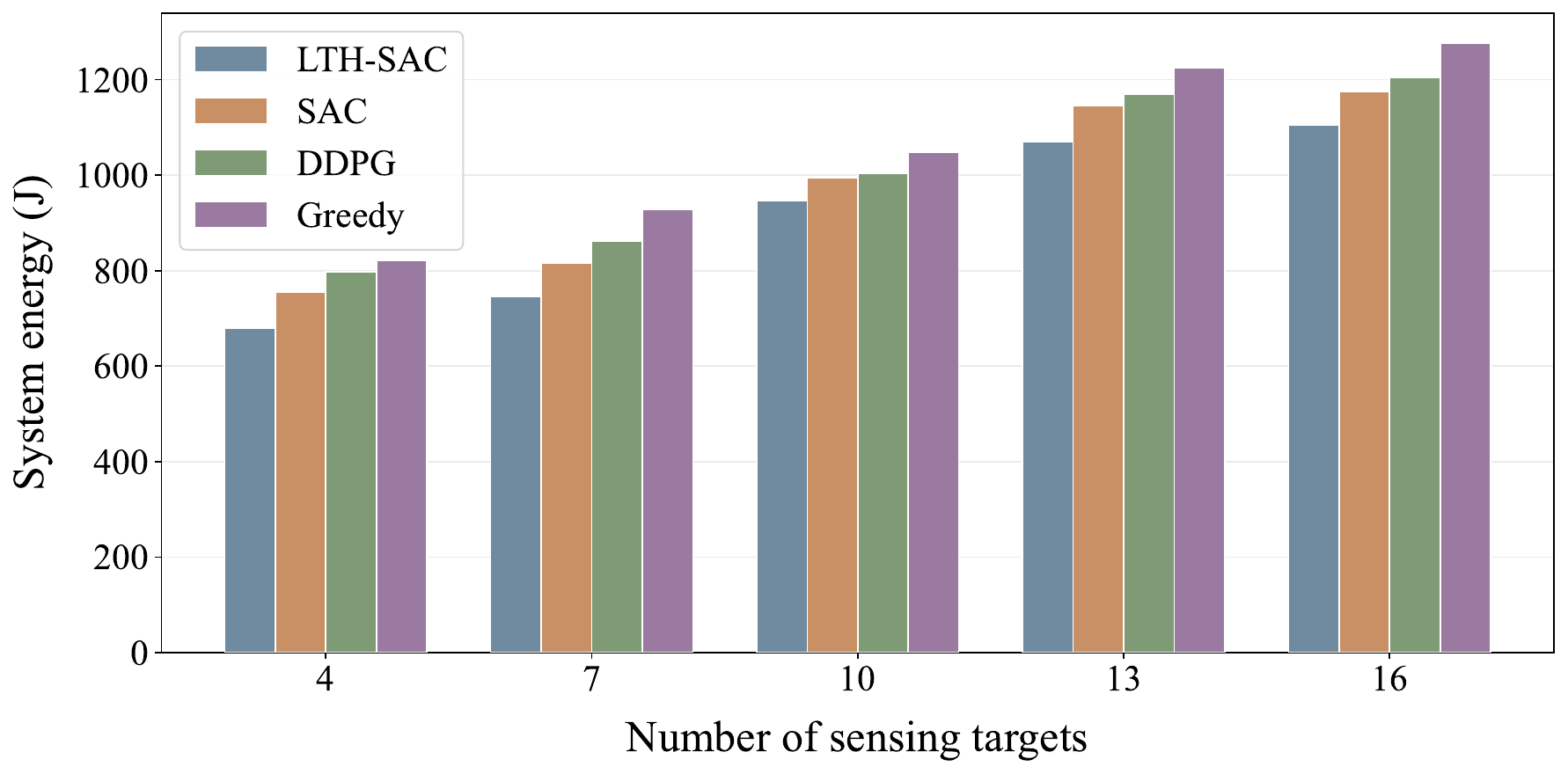}
		\label{fig:energy_targets}
	} 
	\par
	\subfloat[Different numbers of UAVs.]
	{
		\includegraphics[width=0.8\linewidth]{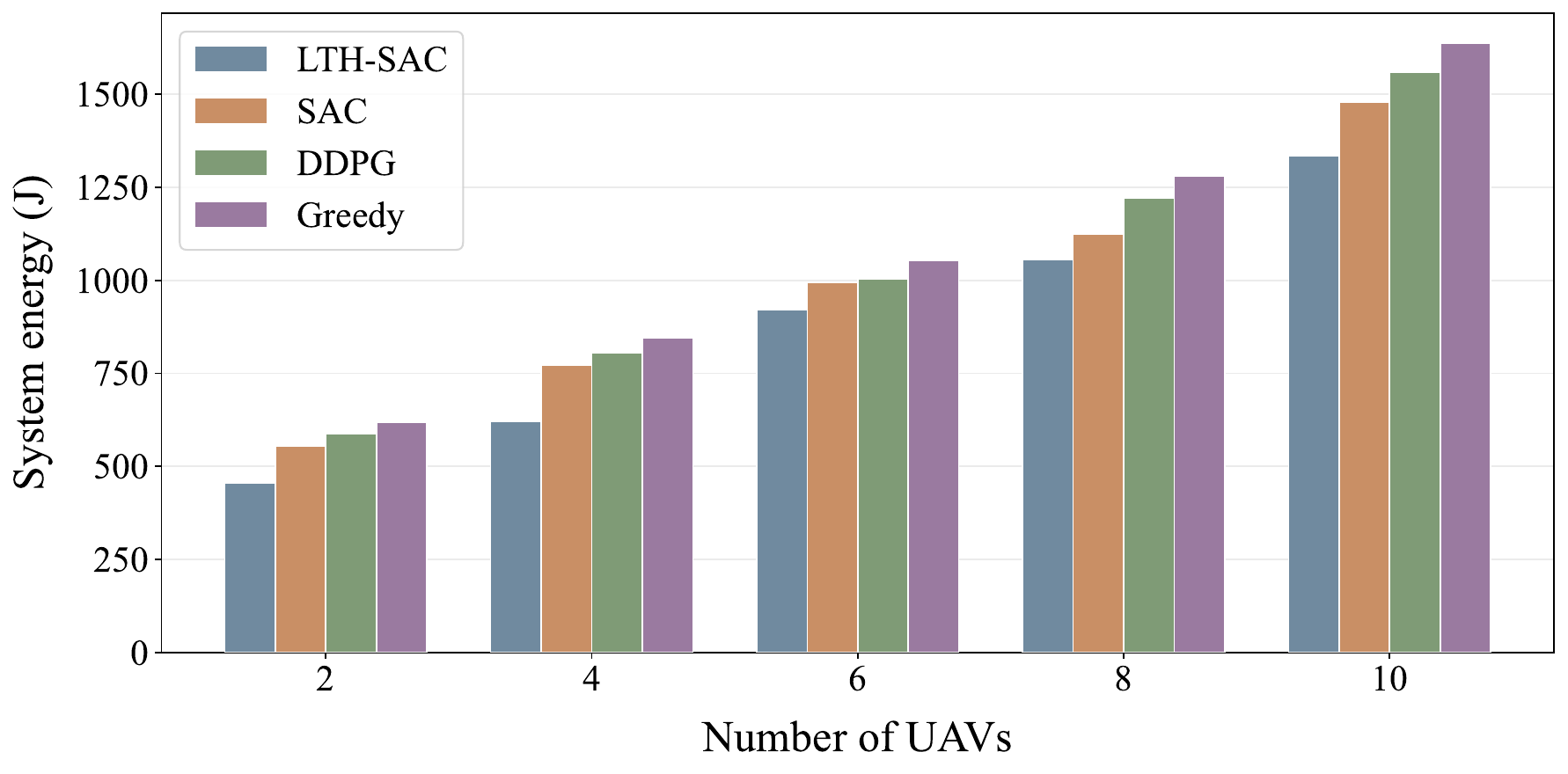}
		\label{fig:energy_uavs}
	}
	\caption{System energy under different network scales.}
	\label{fig:scalability}
	\vspace{-2mm}
\end{figure}

Fig.~\ref{fig:xi} investigates the impact of the average time allocation ratio \(\xi_u\) on sensing data generation, inference accuracy, and UAV queue backlog. As \(\xi_u\) increases, the amount of sensing data increases and the normalized inference error decreases, while the normalized UAV queue backlog grows. This is because allocating more time to sensing generates more data while reducing the time available for communication, thereby increasing queueing pressure. These results highlight the inherent trade-off between sensing performance and queue stability under limited sensing and communication resources.
Fig.~\ref{fig:lambda_energy_tradeoff} and Fig.~\ref{fig:lambda_queue_tradeoff} further examine the impact of the computation offloading ratio \(\lambda_u\) on energy consumption and queue backlogs. As \(\lambda_u\) increases, more sensed data are offloaded to the USV, reducing UAV local-computing energy while increasing USV computing energy. The total system energy does not vary monotonically with \(\lambda_u\), indicating that excessive offloading does not necessarily minimize the overall energy consumption due to the redistribution of computation and communication loads. Moreover, both the UAV-side and USV-side queue backlogs increase with \(\lambda_u\), showing that aggressive offloading may transfer computation burden to the communication and the USV processing. Therefore, an intermediate offloading level is generally required to balance UAV energy savings, USV computing cost, and queue stability.

Fig.~\ref{fig:V} illustrates the impact of the Lyapunov control parameter \(V\) on the average system energy consumption of different algorithms. As \(V\) increases, the energy consumption of LTH-SAC, SAC, and DDPG decreases, whereas that of the greedy algorithm remains almost unchanged. This trend is consistent with the drift-plus-penalty principle, since a larger $V$ assigns greater importance to energy minimization in the resource-allocation objective. Moreover, LTH-SAC consistently achieves the lowest energy consumption over the entire range of $V$, demonstrating that the proposed hierarchical policy can more effectively coordinate target association, UAV trajectory, time allocation, computation offloading, and communication resources. 

Fig.~\ref{fig:scalability} evaluates the scalability of different algorithms under varying network scales. As shown in Fig.~\ref{fig:energy_targets}, the total energy consumption generally increases with the number of sensing targets because the UAVs must undertake heavier sensing, communication, and computation workloads. Similarly, Fig.~\ref{fig:energy_uavs} shows that increasing the number of UAVs results in higher propulsion, sensing, computation, and transmission energy consumption. Nevertheless, LTH-SAC consistently achieves the lowest energy consumption in both scenarios. These results demonstrate that LTH-SAC can effectively coordinate UAV trajectories and resource allocation while maintaining superior energy efficiency under different sensing workloads and network scales, thereby validating its adaptability and scalability.

\begin{figure*}[t]
	\centering
	
	\subfloat[Target distribution I.]
	{
		\includegraphics[width=0.3\textwidth, trim=0pt 0pt 0pt 0pt, clip]{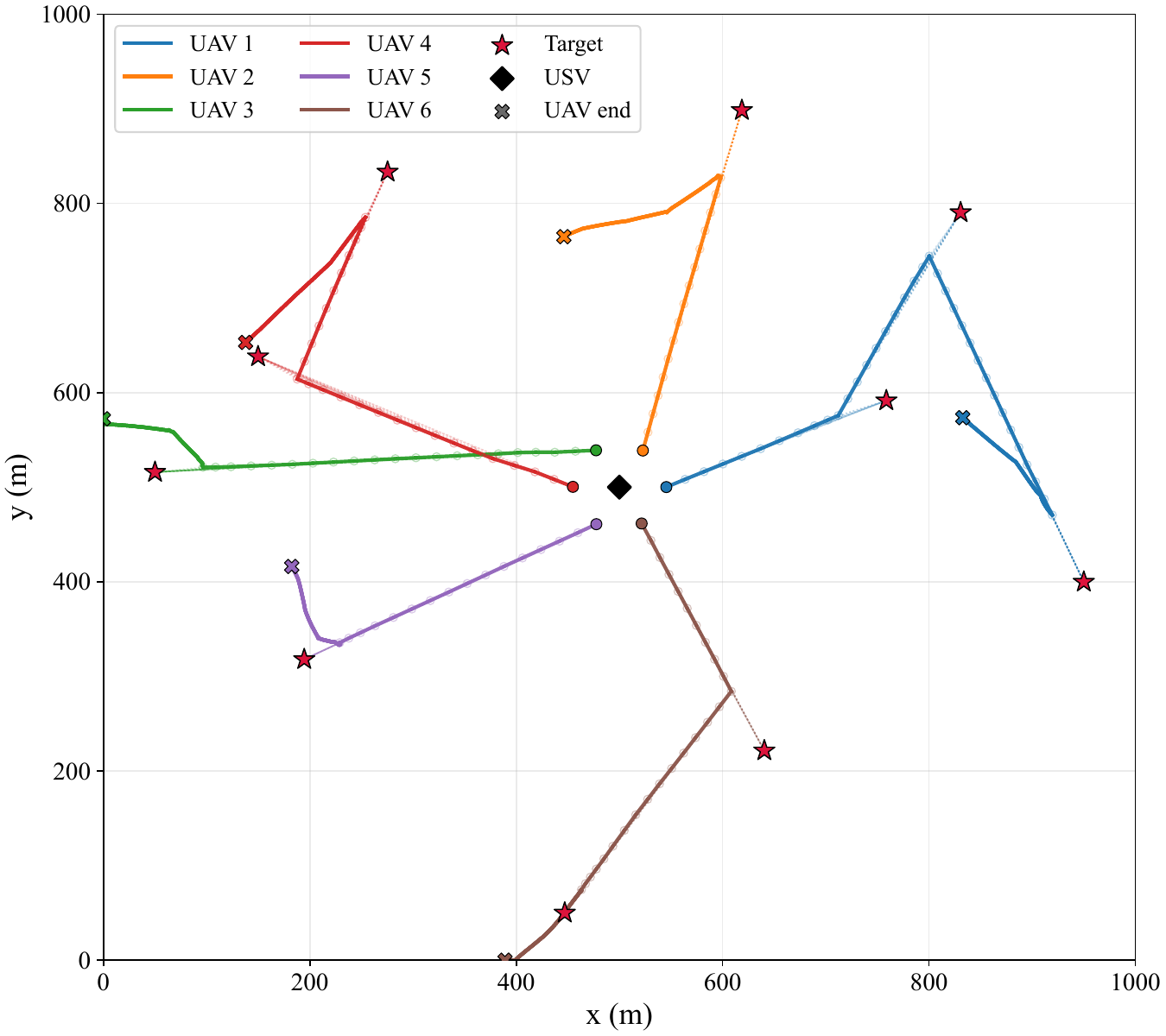}
		\label{fig:trajectory_a}
	}
	\hfill
	\subfloat[Target distribution II.]
	{
		\includegraphics[width=0.3\textwidth, trim=0pt 0pt 0pt 0pt, clip]{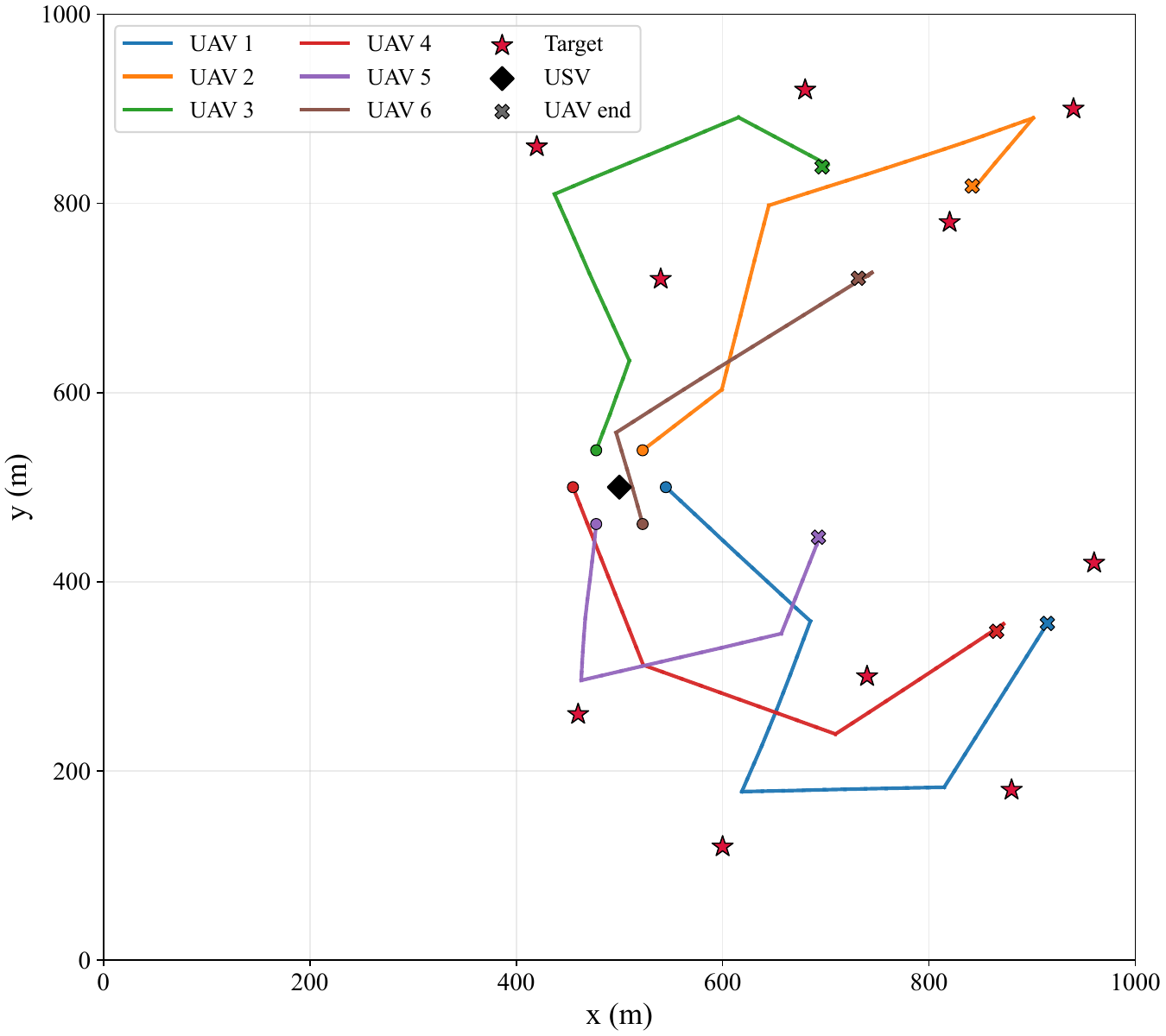}
		\label{fig:trajectory_b}
	}
	\hfill
	\subfloat[Target distribution III.]
	{
		\includegraphics[width=0.3\textwidth, trim=0pt 0pt 0pt 0pt, clip]{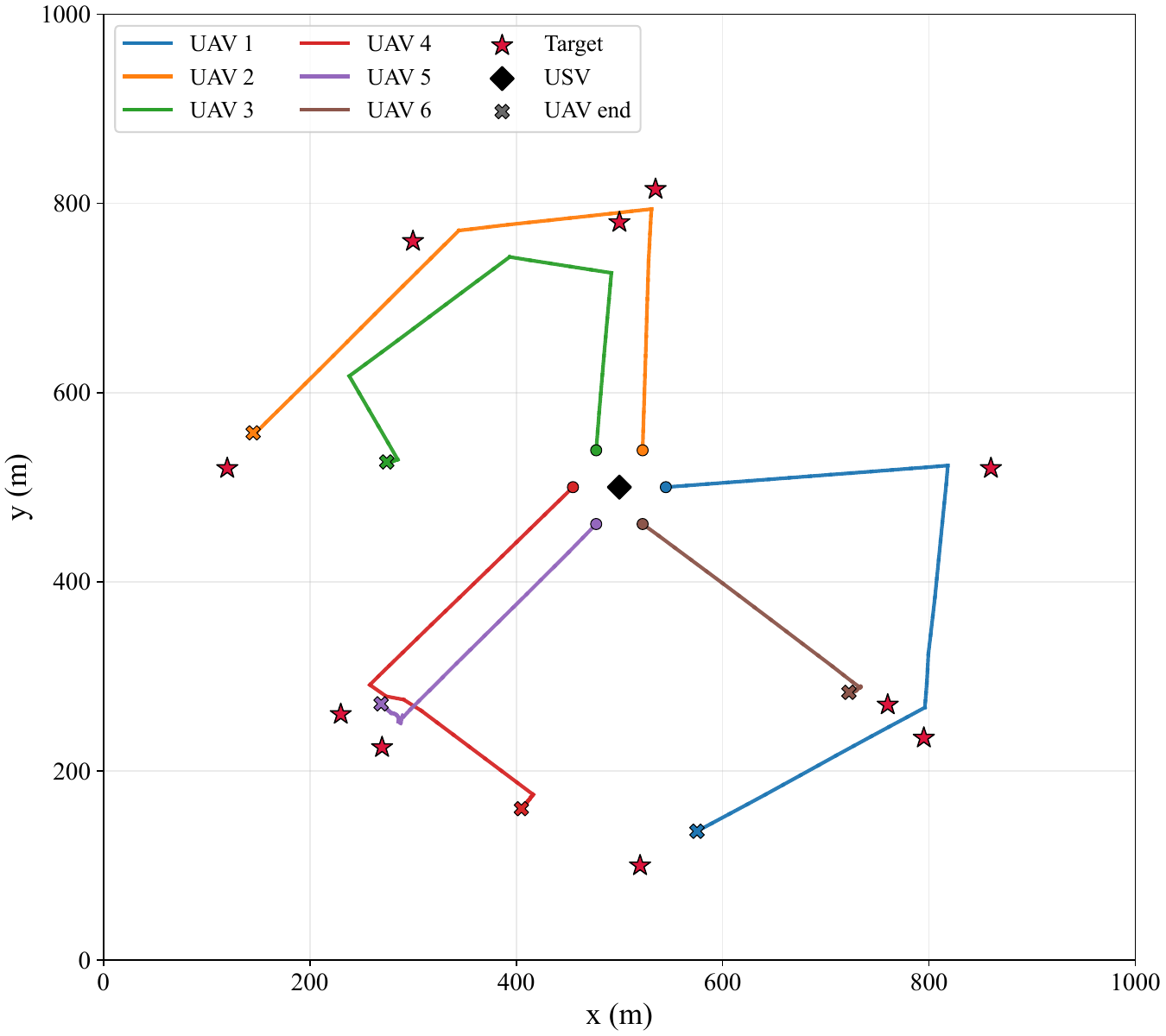}
		\label{fig:trajectory_c}
	}
	
	\caption{Optimized UAV trajectories under different target spatial distributions.}
	\label{fig:trajectory}
	
	\vspace{-2mm}
\end{figure*}

Fig.~\ref{fig:trajectory} shows the optimized UAV trajectories under different target spatial distributions. The proposed LTH-SAC adaptively adjusts UAV-target association and trajectory planning according to target locations, directing different UAVs toward different sensing regions while avoiding unnecessary wide-area exploration. As the target distribution changes, the UAV trajectories are reconfigured accordingly, demonstrating the adaptability of the high-level policy to heterogeneous maritime sensing scenarios.

\section{Conclusion}
\label{sec:conclusion}
This paper investigated the joint UAV trajectory planning and resource allocation in a maritime ISCC system, considering maritime environmental effects and coupled UAV–USV queue dynamics. We formulated a dual-timescale stochastic optimization problem and proposed an LTH-SAC framework to minimize long-term energy consumption while maintaining system performance and queue stability. Simulation results demonstrated that LTH-SAC achieves lower energy consumption and better queue stability than benchmark schemes under different network scales. Moreover, allocating more time to sensing improves sensing performance but reduces communication opportunities and increases UAV queue backlogs, while greater computation offloading reduces UAV local computing energy at the expense of increased USV processing load and queueing pressure. These findings highlight the effectiveness of queue-aware resource management in balancing sensing, communication, computation, and system stability in maritime ISCC systems.


\begin{thebibliography}{1}
\bibliographystyle{IEEEtran}

\bibitem{YWang2026-1}
Y. Wang, R. Li, B. Tian, B. Picano, Y. Chen, and G. Han,
``Queue-aware deep reinforcement learning for UAV-assisted maritime
integrated sensing and communication,''
in \emph{Proc. IEEE 37th Int. Symp. Pers., Indoor Mobile Radio Commun.
	(PIMRC)}, Singapore, Sep.~2026, pp. 1--6.

\bibitem{FSAlqurashi2023}
F. S. Alqurashi, A. Trichili, N. Saeed, B. S. Ooi, and M.-S. Alouini, ``Maritime communications: A survey on enabling technologies, opportunities,
and challenges,'' \emph{IEEE Internet Things J.}, vol. 10, no. 4, pp. 3525–3547, Feb. 2023.


\bibitem{FLiu2022}
F.~Liu, Y.~Cui, C.~Masouros, J.~Xu, T.~X.~Han, Y.~C.~Eldar, and S.~Buzzi, ``Integrated sensing and communications: Toward dual-functional wireless networks for 6G and beyond,'' \textit{IEEE J. Sel. Areas Commun.}, vol.~40, no.~6, pp.~1728--1767, Jun.~2022.

\bibitem{FLiu2020}
R. Li, Z. Xiao, and Y. Zeng,
``Toward seamless sensing coverage for cellular multi-static
integrated sensing and communication,''
\emph{IEEE Trans. Wireless Commun.}, vol. 23, no. 6,
pp. 5363--5376, Jun. 2024.

\bibitem{SLiu2020}
S.~Lu \textit{et al.}, ``Integrated sensing and communications: Recent advances and ten open challenges,'' \textit{IEEE Internet Things J.}, vol.~11, no.~11, pp.~19094--19120, Jun.~2024.

\bibitem{Mach2017Mobile}
P.~Mach and Z.~Becvar,
``Mobile edge computing: A survey on architecture and computation offloading,''
\emph{IEEE Commun. Surveys Tuts.}, vol.~19, no.~3,
pp.~1628--1656, Third Quart.~2017.

\bibitem{Mao2017MEC}
Y.~Mao, C.~You, J.~Zhang, K.~Huang, and K.~B.~Letaief,
``A survey on mobile edge computing: The communication perspective,''
\emph{IEEE Commun. Surveys Tuts.}, vol.~19, no.~4,
pp.~2322--2358, Fourth Quart.~2017.

\bibitem{Lin2024Maritime}
Z.~Lin, J.~Yang, Y.~Chen, C.~Xu, and X.~Zhang,
``Maritime distributed computation offloading in space-air-ground-sea integrated networks,''
\emph{IEEE Commun. Lett.}, vol.~28, no.~7,
pp.~1614--1618, Jul.~2024.

\bibitem{Wen2025Survey}
D.~Wen, Y.~Zhou, X.~Li, Y.~Shi, K.~Huang, and K.~B.~Letaief,
``A survey on integrated sensing, communication, and computation,''
\emph{IEEE Commun. Surveys Tuts.}, vol.~27, no.~5,
pp.~3058--3098, Oct.~2025.

\bibitem{CLi2025}
C. Li, M. Dong, Y. Fu, F. Richard Yu and N. Cheng, ``Integrated sensing, communication, and computation for IoV: Challenges and opportunities,'' \emph{IEEE Commun. Surveys Tuts.}, vol. 28, pp. 1136-1168, Sep. 2025.

\bibitem{WChen2026}
W. Chen, Y. He, Z. Ye, D. Wen, and G. Yu,
``Task-oriented multimodal edge intelligence via
integrated sensing-communication-computation,''
\emph{IEEE Trans. Wireless Commun.}, vol. 25,
pp. 20071--20085, Jul. 2026.


\bibitem{Meng2024UAVEnabled}
K.~Meng, Q.~Wu, J.~Xu, W.~Chen, Z.~Feng, R.~Schober, and A.~L.~Swindlehurst,
``UAV-enabled integrated sensing and communication: Opportunities and challenges,''
\emph{IEEE Wireless Commun.}, vol.~31, no.~2, pp.~97--104,  Apr.~2024.

\bibitem{Zhou2018Computation}
M. W. Akhtar and N. Saeed,
``UAVs-enabled maritime communications: Opportunities and challenges,''
\emph{IEEE Syst. Man Cybern. Mag.}, vol. 9, no. 3,
pp. 2--8, Jul. 2023.

\bibitem{Liu2023MLTrajectory}
X. Liu, Y. Liu, and Y. Chen,
“Machine learning empowered trajectory and resource optimization for UAV-enabled integrated sensing, communication, and computation,”
\emph{IEEE Trans. Wireless Commun.},
vol. 22, no. 11, pp. 7613--7627, Nov. 2023.

\bibitem{ZWang2024}
Z. Wang, J. Du, C. Jiang, Y. Ren, and X.-P. Zhang,
``UAV-assisted target tracking and computation offloading
in USV-based MEC networks,''
\emph{IEEE Trans. Mobile Comput.}, vol. 23, no. 12,
pp. 11389--11405, Dec. 2024.

\bibitem{RZhang2026}
R. Zhang, F. Dong, and W. Wang,
``UAV-USV collaborative sensing system based on
integrated sensing and communication,''
\emph{IEEE Sensors J.}, early access, Jan. 2026.


\bibitem{CarreteroMoya2010SeaClutter}
J.~Carretero-Moya, J.~Gismero-Menoyo, A.~Blanco-del-Campo, and A.~Asensio-Lopez,
``Statistical analysis of a high-resolution sea-clutter database,''
\emph{IEEE Trans. Geosci. Remote Sens.}, vol.~48, no.~4,
pp.~2024--2037, Apr.~2010.

\bibitem{Nomikos2023Survey}
N. Nomikos, P. K. Gkonis, P. S. Bithas, and P. Trakadas,
``A survey on UAV-aided maritime communications:
Deployment considerations, applications, and future challenges,''
\emph{IEEE Open J. Commun. Soc.}, vol. 4, pp. 56--78, 2023.




\bibitem{Soltani2011ShipAntenna}
M.~N.~Soltani, R.~Izadi-Zamanabadi, and R.~Wisniewski,
``Reliable control of ship-mounted satellite tracking antenna,''
\emph{IEEE Trans. Control Syst. Technol.}, vol.~19, no.~1,
pp.~221--228, Jan.~2011.

\bibitem{PHou2025}
P. Hou et al.,
``Distributed DRL-based integrated sensing, communication,
and computation in cooperative UAV-enabled intelligent
transportation systems,''
\emph{IEEE Internet Things J.}, vol. 12, no. 5,
pp. 5792--5806, Mar.~2025.

\bibitem{Deng2024MultiUAVISCC}
C.~Deng, X.~Fang, and X.~Wang,
``Integrated sensing, communication, and computation with adaptive DNN splitting in multi-UAV networks,''
\emph{IEEE Trans. Wireless Commun.}, vol.~23, no.~11,
pp.~17429--17445, Nov.~2024.

\bibitem{Zhou2024UAVISCC}
Y. Zhou, X. Liu, X. Zhai, Q. Zhu, and T. S. Durrani,
``UAV-enabled integrated sensing, computing, and communication for Internet of Things: Joint resource allocation and trajectory design,''
IEEE Internet Things J., vol. 11, no. 7, pp. 12717--12727, Apr. 2024.

\bibitem{Chen2024UAVISCC}
J. Chen, Y. Xu, D. Yang, and T. Zhang,
``UAV-assisted ISCC networks: Joint resource and trajectory optimization,''
\emph{IEEE Wireless Commun. Lett.}, vol. 13, no. 9, pp. 2372--2376, Sep. 2024.

\bibitem{Wang2018Wireless}
J.~Wang \emph{et al.},
``Wireless channel models for maritime communications,''
\emph{IEEE Access}, vol.~6, pp.~68070--68088, Nov. 2018.

\bibitem{Zhang2025U2SChannel}
C. Zhang, M. Yang, L. Bai, B. Ai, R. He, Z. Gao, Y. Gong, and G. Shi,
``A geometry-based marine channel model for UAV-to-ship communication systems,''
\emph{IET Commun.}, vol. 19, no. 1, Art. no. e70104, Nov. 2025.

\bibitem{NAhmed2025}
N. Ahmed, F. Ali, Q. Deng, Q. Zhu, X. Chen, B. Hua, J. Bao, and K. Mao,
``Impact of sea cluttering and wave shadowing on U2S MIMO
channel model incorporating UAV-ship 6D motion in
maritime environments,''
\emph{Veh. Commun.}, vol. 55, Art. no. 100963, Oct. 2025.

\bibitem{JYu2024}
J. Yu, C. Hua, and L. Liu,
``Adaptive beamforming for unmanned surface vehicle-to-satellite
communication in harsh sea environments,''
\emph{Mobile Communications}, vol. 48, no. 11,
pp. 70--76, Nov. 2024.

\bibitem{Zhang2024MaritimeISAC}
J. Zhang, G. Wang, H. Yang, B. Liu, and B. Li,
``Multiuser maritime integrated sensing and communication
shipboard base station deployment optimization,''
\emph{IEEE Internet Things J.}, vol. 11, no. 18,
pp. 29375--29386, Jun. 2024.

\bibitem{Zhang2024MaritimeResource}
J. Zhang, G. Wang, H. Yang, B. Liu, and B. Li,
``Multi-user maritime integrated sensing and communication resource allocation in auxiliary communication scenarios,''
\emph{J. Commun. Netw.}, vol. 26, no. 5, pp. 566--576, Oct. 2024.

\bibitem{Li2025MaritimeISAC}
B. Li, J. Liu, Y. Xiong, J. Mu, P. Xiao, and S. Chen,
``UAV-enabled integrated sensing and communication in maritime emergency networks,''
\emph{IEEE Internet Things J.}, vol. 12, no. 24, pp. 53997--54012, Dec. 2025.

\bibitem{Dai2022SmartOceans}
M. Dai, Y. Li, P. Li, Y. Wu, L. Qian, B. Lin, and Z. Su,
``A survey on integrated sensing, communication, and computing networks for smart oceans,''
\emph{J. Sens. Actuator Netw.}, vol. 11, no. 4, Art. no. 70, Oct.~2022.

\bibitem{Xu2023UAVISCC}
Y. Xu, T. Zhang, Y. Liu, and D. Yang,
``UAV-enabled integrated sensing, computing, and communication: A fundamental trade-off,''
\emph{IEEE Wireless Commun. Lett.}, vol. 12, no. 5, pp. 843--847, May 2023.

\bibitem{Liu2025ISCCOffloading}
P. Liu, Z. Fei, X. Wang, Y. Zhou, Y. Zhang, and F. Liu,
``Joint beamforming and offloading design for integrated sensing, communication, and computation system,''
\emph{IEEE Trans. Veh. Technol.}, vol. 74, no. 9, pp. 14933--14937, Sep. 2025.

\bibitem{Zhang2019UAVMEC}
J.~Zhang, L.~Zhou, Q.~Tang, E.~C.-H.~Ngai, X.~Hu, H.~Zhao, and J.~Wei,
``Stochastic computation offloading and trajectory scheduling for UAV-assisted mobile edge computing,''
\emph{IEEE Internet Things J.}, vol.~6, no.~2,
pp.~3688--3699, Apr.~2019.




\bibitem{Bai2025Dynamic}
Y. Bai and Y. Zhang,
``Dynamic offloading based on Lyapunov optimization for
UAV-assisted maritime IoT-MEC networks,''
\emph{IEEE Trans. Veh. Technol.}, vol. 74, no. 11,
pp. 17894--17906, Nov.~2025.

\bibitem{Gao2026ISCCQueue}
X. Gao, X. Hu, W. Wang, K.-K. Wong, K. Yang, and C.-B. Chae,
``Integrated sensing, communications, and computation in edge-intelligent networks: An online resource management approach,'' \emph{IEEE Trans. Wireless Commun.}, vol. 25, pp. 16578--16595, May~2026.

\bibitem{Li2026MaritimeISACQueue}
B. Li, M. Ye, H. Liu, Y. Gong, N. Gao, J. Nie, P. Xiao, and X. Cheng,
``Queue-aware graph reinforcement learning for UAV-ISAC-assisted maritime data collection,''
\emph{arXiv:2607.00324}, Jul. 2026.

\bibitem{Gou2026UnderwaterISACQueue}
Y. Gou, T. Zhang, J. Liu, Y. Bi, and G. Han,
``Providing efficient communication for underwater acoustic sensor networks with ISAC: A deep multi-agent reinforcement learning approach,''
\emph{IEEE Trans. Mobile Comput.}, early access, Sep. 2026.







\bibitem{chen2009mimo}
C. Y. Chen and P. P. Vaidyanathan, ``MIMO radar waveform optimization with prior information of the extended target and clutter,'' \emph{IEEE Trans. Signal Process.}, vol. 57, no. 9, pp. 3533--3544, Sep. 2009.


\bibitem{Lorenz2005MVDR}
R.~G.~Lorenz and S.~P.~Boyd,
``Robust minimum variance beamforming,''
\emph{IEEE Trans. Signal Process.}, vol.~53, no.~5,
pp.~1684--1696, May~2005.

\bibitem{SPeng2025}
S. Peng, B. Li, L. Liu, Z. Fei, and D. Niyato, ``Trajectory design and resource allocation for multi-UAV-assisted sensing, communication, and edge computing integration,'' \emph{IEEE Trans. Commun.}, vol. 73, no. 4, pp. 2847–2861, Apr. 2025.

\bibitem{BLi2025}
B. Li, B. Liu, H. Yang, G. Wang, J. Zhang, and N. Zhao,
``Waveform design of UAV-enabled integrated sensing and
communication in marine environment,''
\emph{J. Electron. Inf. Technol.}, vol. 47, no. 4,
pp. 958--967, Apr. 2025.

\bibitem{CZeng2021}
C. Zeng, J.-B. Wang, C. Ding, H. Zhang, M. Lin, and J. Cheng,
``Joint optimization of trajectory and communication resource
allocation for unmanned surface vehicle enabled maritime
wireless networks,''
\emph{IEEE Trans. Commun.}, vol. 69, no. 12,
pp. 8100--8115, Dec. 2021.

\bibitem{Zhang2022Accelerating}
T.~Zhang, S.~Wang, G.~Li, F.~Liu, G.~Zhu, and R.~Wang,
``Accelerating edge intelligence via integrated sensing and communication,''
in \emph{Proc. IEEE Int. Conf. Commun. (ICC)},
Seoul, South Korea, May~2022, pp.~1586--1592.


\bibitem{Zeng2019UAVEnergy}
Y.~Zeng, J.~Xu, and R.~Zhang,
``Energy minimization for wireless communication with rotary-wing UAV,''
\emph{IEEE Trans. Wireless Commun.}, vol.~18, no.~4,
pp.~2329--2345, Apr.~2019.

\bibitem{Bi2021Lyapunov}
S.~Bi, L.~Huang, H.~Wang, and Y.-J.~A.~Zhang,
``Lyapunov-guided deep reinforcement learning for stable online computation offloading in mobile-edge computing networks,''
\emph{IEEE Trans. Wireless Commun.}, vol.~20, no.~11,
pp.~7519--7537, Nov.~2021.

\bibitem{Xu2022MultiUAV}
Y.~Xu, T.~Zhang, Y.~Liu, D.~Yang, L.~Xiao, and M.~Tao, ``Cellular-connected multi-UAV MEC networks: An online stochastic optimization approach,'' \emph{IEEE Trans. Commun.}, vol.~70, no.~10, pp.~6630--6647, Oct.~2022.


\end{thebibliography}
\end{document}